\documentclass[11pt]{article}
\usepackage[margin=1in]{geometry}
\usepackage{algorithmicx}

\usepackage{pifont}

\usepackage[T1]{fontenc}
\usepackage[utf8]{inputenc}
\usepackage{authblk}
\usepackage{amsmath, amssymb, amsthm, thmtools, amsfonts, bm, bbm, thm-restate}
\usepackage{algorithm}
\usepackage{algorithmicx}
\usepackage[noend]{algpseudocode}
\usepackage[noadjust]{cite}
\usepackage{pdfpages}

\usepackage{asymptote}
\usepackage{graphicx}
\usepackage{todonotes}
\usepackage{multirow}
\usepackage{comment}
\usepackage{dsfont}
\usepackage{bigstrut}
\usepackage{caption}
\usepackage{subcaption}
\usepackage{multirow}
\usepackage{makecell}
\usepackage{booktabs}
\usepackage{xcolor}

\usepackage{commath}

\definecolor{BrickRed}{rgb}{0.8,0.25,0.33}
\usepackage[pagebackref,colorlinks,citecolor=blue,linkcolor=BrickRed]{hyperref}

\usepackage{enumitem}
\usepackage[normalem]{ulem}

\usepackage[framemethod=tikz]{mdframed}
\usepackage{nicefrac}

\usepackage{tikz}
\usetikzlibrary{positioning, fit, calc}

\usepackage{pifont}

\usepackage{cleveref}

\newtheorem{thm}{Theorem}[section]

\newtheorem{prop}[thm]{Proposition}

\newtheorem{fact}[thm]{Fact}
\newtheorem{lem}[thm]{Lemma}

\newtheorem{Def}[thm]{Definition}
\newtheorem{obs}[thm]{Observation}

\newtheorem{rem}[thm]{Remark}

\crefname{thm}{Theorem}{Theorems}
\crefname{cla}{Claim}{Claims}
\crefname{lem}{Lemma}{Lemmas}
\crefname{fact}{Fact}{Facts}
\crefname{prop}{Proposition}{Propositions}

\allowdisplaybreaks

\newcommand{\E}{\mathbb{E}}
\newcommand{\R}{\mathbb{R}}
\renewcommand{\Pr}{\mathds{P}r}

\newcommand{\eps}{\varepsilon}
\newcommand{\calA}{\mathcal{A}}
\newcommand{\calB}{\mathcal{B}}

\newcommand{\calF}{\mathcal{F}}

\newcommand{\calI}{\mathcal{I}}
\newcommand{\calK}{\mathcal{K}}
\newcommand{\calM}{\mathcal{M}}
\newcommand{\calP}{\mathcal{P}}

\newcommand{\poly}{\mathrm{poly}}
\newcommand{\Exp}{\mathrm{Exp}}

\newcommand{\Uni}{\mathrm{Uni}}

\newcommand{\Cov}{\mathrm{Cov}}

\newcommand{\calC}{\mathcal{C}}

\newenvironment{wrapper}[1]
{
	\smallskip
	\begin{center}
		\begin{minipage}{\linewidth}
			\begin{mdframed}[hidealllines=true, backgroundcolor=gray!20, leftmargin=0cm,innerleftmargin=0.275cm,innerrightmargin=0.275cm,innertopmargin=0.275cm,innerbottommargin=0.275cm,roundcorner=10pt]
				#1}
			{\end{mdframed}
		\end{minipage}
	\end{center}
	\smallskip
}

\let\vec\mathbf
\renewcommand{\vec}{\mathbf}

    \author[1]{Niv Buchbinder\thanks{Supported in part by ISF grant 3001/24 and United States - Israel
BSF grant 2022418.}}
    \author[2]{Joseph (Seffi) Naor\thanks{Supported in part by ISF grant 3001/24 and United States - Israel
BSF grant 2022418.}}
    \author[2]{David Wajc\thanks{Supported in part by a Taub Family Foundation ``Leader in Science and Technology'' fellowship and ISF grant 3200/24.}}

    \affil[1]{Tel Aviv University}
    \affil[2]{Technion --- Israel Institute of Technology}

\title{Chasing Submodular Objectives,\\and Submodular Maximization via Cutting Planes}

\date{\vspace{-1cm}}

\begin{document}

\maketitle

\pagenumbering{gobble}
\begin{abstract}
We introduce the \emph{submodular objectives chasing problem}, which generalizes many natural and previously-studied problems:~a sequence of constrained submodular maximization problems is revealed over time, with both the objective and available ground set changing at each step.
The goal is to maintain solutions of high approximation and low total \emph{recourse} (number of changes), compared with exact offline algorithms for the same input sequence.
For the central cardinality constraint and partition matroid constraints we provide polynomial-time algorithms achieving both optimal $(1-1/e-\eps)$-approximation and optimal competitive recourse for \emph{any} constant-approximation.

\medskip 

Key to our algorithm's polynomial time, and of possible independent interest, is a new meta-algorithm for $(1-1/e-\eps)$-approximately maximizing the multilinear extension under general constraints, which we call {\em approximate-or-separate}. Our algorithm relies on an improvement of the round-and-separate method [Gupta-Levin SODA'20], inspired by an earlier proof by [Vondr\'ak, PhD~Thesis'07].
The algorithm, whose guarantees are similar to the influential {\em continuous greedy} algorithm [Calinescu-Chekuri-P\'al-Vondr\'ak SICOMP'11], can use any cutting plane method and separation oracle for the constraints.
This allows us to introduce cutting plane methods, used for exact unconstrained submodular minimization since the '80s [Grötschel/Lovász/Schrijver Combinatorica'81], as a useful method for (optimal approximate) constrained submodular maximization.
We show further applications of this approach to static algorithms with curvature-sensitive approximation, and to communication complexity protocols.
\end{abstract}

\tableofcontents

\newpage

\pagenumbering{arabic}
\section{Introduction}\label{sec:intro}

Constrained submodular maximization is a rich area of combinatorial optimization, studied for at least half a century \cite{nemhauser1978analysis,nemhauser1978best}.\footnote{A set function $f:2^E\to \mathbb{R}_+$ is \emph{submodular} if it satisfies for all $S\subseteq T\subseteq E\setminus \{i\}$, it satisfies $f(i \mid S) \geq f(i \mid T)$, for $f(i \mid S):=f(S\cup \{i\}) - f(S)$ the \emph{marginal} gain of adding $i$ to $S$. It is also \emph{monotone} if for all $S\subseteq T$, $f(S)\leq f(T)$.} 
Motivated by such problems' wide applicability, numerous different algorithmic paradigms have been developed for such problems in static settings \cite{nemhauser1978analysis,calinescu2011maximizing,filmus2014monotone,sviridenko2017optimal,BF24,BF24b,BF25}, streaming algorithms \cite{alaluf2022optimal,chekuri2015streaming,feldman2022streaming,feldman2020oneway,kazemi2019submodular,levin2021streaming}, online and secretary settings~\cite{bateni2013submodular,buchbinder2020online,buchbinder2019online,feldman2018goes}, parallel computation \cite{balkanski2022adaptive2,balkanski2018adaptive1,chekuri2019parallel1,chekuri2019parallel2}, distributed computing~\cite{mirzasoleiman2016distributed,pontebarbosa2016new}, and more.
In this work, we search for algorithms applicable to several computational models, but focus mainly on dynamic algorithms with bounded \emph{recourse}.

Minimizing movement cost has a long history in the online and dynamic algorithms literature. Here, inputs change over time and we must update high-quality solutions, while minimizing changes to the solution, known as \emph{recourse}. See, e.g.,  \cite{manasse1988competitive,fiat1991competitive,borodin1992optimal,friedman1993convex,koutsoupias1995k,gupta2014maintaining,bansal2015polylogarithmic,megow2016power,feldkord2018fully,sellke2020chasing,gupta2020fully,argue2021chasing,bhattacharya2021online,fichtenberger2021consistent,gupta2022online,bubeck2023randomized,bhattacharya2023chasing,lkacki2024fully,bhattacharya2025fully}. A growing body of work has yielded results for recourse minimization for \emph{submodular maximization} \cite{chakrabarti2015submodular,chekuri2015streaming,chan2018online,buchbinder2019online,lattanzi2020fully,chen2022complexity,banihashem2023dynamic,banihashem2024dynamic,dutting2024consistent,dutting2025cost}.
In this work we initiate the study of recourse-minimizing dynamic algorithms for constrained submodular maximization \emph{with evolving submodular objectives}.
\vspace{-0.15cm}
\begin{wrapper}\textbf{The Submodular Objectives Chasing Problem:}\\ 
Initially, we are given a ground set $E$ and family of feasible subsets $\calI\subseteq 2^E$.
Then, at each time step $t=1, \ldots, T$ we are revealed a subset of available elements $E_t \subseteq E$, a non-negative monotone submodular function $f_t:2^{E_t}\to~\mathbb{R}_+$, 
and a desired value $V_t \leq OPT_t:= \max_{S \in {\cal I}\cap 2^{E_t}}f_t(S)$. 
Then, before the next time step, we must output a feasible subset $S_t\in {\cal I}\cap 2^{E_t}$ of available elements that is \emph{$\alpha$-approximate}, $\E[f_t(S_t)]\geq \alpha\cdot V_t$, with the goal of minimizing the \emph{recourse}, $\sum_{t=1}^{T} \E[|S_t \oplus S_{t-1}|]$, for $A \oplus B := (A\setminus B)\cup (B\setminus A)$ the symmetric difference.
\end{wrapper}
\vspace{-0.15cm}

Due to the ubiquity of submodularity, this abstract problem captures various concrete problems. We outline two motivating applications here. In both applications the underlying constraint is (or can be cast as) a cardinality or partition constraint.\footnote{A \emph{partition constraint} is defined by a partition of the ground set $E$ into  parts $P_1,\dots,P_r$ with per-part \emph{cardinality constraints} $k_i$ for each $i$, i.e., the feasible sets are $\calI:=\{S\subseteq E \mid |S\cap P_i|\leq k_i$ for all $i\in [r]$\}. It is \emph{simple} if $k_i=1\forall i$.} These are the best understood and most widely used constraints, and are also the focus of our results for submodular objective chasing.
\vspace{-0.4cm}
\paragraph{Fully-dynamic submodular welfare:}
Consider a market consisting of buyers and items, both arriving and departing over time. Each item may be allocated to at most one buyer. Each buyer $i$ has a non-negative monotone submodular valuation function $f_i$. At each time, we wish to maintain an allocation of the currently present items to the present buyers, of near-maximum \emph{social welfare}, $\sum_i f_i(S_i)$, for $S_i$ the items allocated to buyer $i$, while minimally changing the allocation over time. (Changes incur costs and buyer dissatisfaction.)
Via a standard reduction, this problem fits the submodular chasing problem with a simple partition constraint.
\vspace{-0.4cm}
\paragraph{Fully-dynamic max coverage:}
Consider a max-coverage instance consisting of a ground set of elements and a family of subsets, where both elements and subsets arrive and depart over time. At any time, the algorithm may hold at most $k$ subsets. Our goal is to cover a maximum number of elements while minimally changing the selected subsets. (In applications to facility location, sensor placing, hiring, etc, changes typically incur activation costs.) 
This problem fits the submodular chasing problem with a cardinality constraint. Note that $f_t$, the (submodular) coverage function at time $t$,\footnote{A \emph{coverage} function $f:2^E\to \mathbb{R}_+$ is a function associating each element of $E$ with a set in some family, and taking  $f(S)=|\bigcup_{i\in S} i|$. It is easy to see that such $f$ are non-negative monotone submodular functions.} changes as elements in the ground set of the set-cover instance arrive and depart.

\smallskip 
Previously, numerous works studied special cases of the {\em submodular objective chasing problem}. For example, many results are known for a fixed function $f$ (i.e., $f_t=f$ for all $t$) subject to cardinality constraints in incremental settings, $E_t=E_{t-1}\cup \{e_t\}$, while minimizing worst-case~(i.e.,~per-step) or amortized (i.e., average over time steps) \emph{absolute recourse} (i.e., recourse on the worst-case sequence) \cite{chakrabarti2015submodular,chekuri2015streaming,chan2018online,buchbinder2019online,lattanzi2020fully,chen2022complexity,banihashem2023dynamic,banihashem2024dynamic,dutting2024consistent,dutting2025cost}.
The state-of-the-art includes $(1-1/e-\epsilon)$-approximations using polylogarithmic amortized recourse \cite{chen2022complexity}, and constant worst-case recourse slow $2/3$-approximation and polytime $0.51$-approximation \cite{dutting2025cost}.
(See \Cref{sec:related-work} for more.)

While such absolute recourse guarantees help us characterize approximation/recourse tradeoffs for worst-case sequences, they naturally imply loose bounds for sequences chaseable using $o(T)$ total recourse.
Moreover, for the general problem with cardinality constraint of $k$, non-trivial absolute recourse guarantees are impossible, since successive $\langle f_t, E_t\rangle$ may have disjoint approximate maximizers, forcing $k$ recourse per step.

In this work, we study this general problem using the more fine-grained metric of \emph{competitive recourse}, introduced by  Bhattacharya et al.~\cite{bhattacharya2023chasing}. Here, one aims for near-optimal recourse for any given sequence, \emph{competitive} with that of the optimal \emph{offline} algorithm, which knows the update sequence in advance and can optimize for it.
In our context, we say a submodular objective chasing algorithm is \emph{$c$-competitive} if its recourse is at most $c$ times the minimum recourse needed by an (offline) algorithm to maintain feasible solutions $S_t$ with $f(S_t)\geq V_t$.
This notion of competitiveness is in line with traditional notions of movement cost in online settings,
for problems such as caching \cite{manasse1988competitive,fiat1991competitive}, the $k$-server problem \cite{manasse1988competitive,koutsoupias1995k,bansal2015polylogarithmic,bubeck2023randomized}, metrical task systems \cite{borodin1992optimal,bubeck2023randomized}, and convex body chasing \cite{friedman1993convex,argue2021chasing,sellke2020chasing}. 

Given that polylogarithmic (amortized) absolute recourse is known to be consistent with high approximation ratios for basic cardinality constraints in special cases,  we ask: \emph{what approximation is attainable (efficiently) using polylogarithmic \underline{competitive recourse} under more general dynamics?}

\subsection{Our Main Results}

Our main result is an algorithm achieving logarithmic competitive recourse and $(1-1/e-\eps)$ approximation for chasing submodular objectives subject to either cardinality or partition matroids.

\begin{restatable}{thm}{mainthm}\label{thm-main}
For any $\eps>0$, there is a $(1-1/e-\eps)$-approximate and $ O(\eps^{-1}\log (\max_t |E_t|\eps^{-1}))$-competitive randomized algorithm for monotone submodular objective chasing subject to either cardinality or partition  constraints. 
The algorithm runs in polynomial time, using only value oracle access to the functions~$f_t$.
\end{restatable}
Ignoring recourse, our algorithm's approximation is essentially optimal: a $1-1/e$ approximation is optimal even for the special case of a fixed coverage $f$ function for poly-time algorithms assuming \textsc{P}$\neq $\textsc{NP} \cite{feige1998threshold,moshkovitz2008two,dinur2014analytical} or unconditionally if given only value oracle access to $f$ \cite{nemhauser1978best,mirrokni2008tight}.
   Returning to recourse, as we observe, for constant $\eps>0$ our algorithm's competitive recourse is also asymptotically optimal, and is in fact necessary to achieve \emph{any} constant-approximation.

\begin{restatable}{obs}{chasingLB}\label{lem:chasing-LB}
    No constant-approximate submodular objective chasing algorithm has $o(\log \max_t |E_t|)$-competitive recourse, even for a fixed coverage function $f$ and cardinality constraint of $k=1$.
\end{restatable}

We prove the above in \Cref{app:lower-bounds}, where we also show that our use of randomization is~inevitable: deterministic algorithms require (at least) exponentially higher competitive recourse.

Our competitive recourse result implies absolute (amortized) recourse for sequences of interest, as we show in  \Cref{sec:special-cases}. 
For example, for the well-studied fixed function ($f_t=f$ for all $t$) in the insertion-only (and symmetrically the deletion-only) setting, the amortized recourse needed to maintain a $(1-\eps)$-approximate solution is $O(\eps^{-1}\log (k/\eps))$. 
More interestingly, we show that using the exponential histograms technique from the streaming literature, these results extend to the sliding windows model, where we must output solutions for the last $L$ arriving elements, which is a popular model often used as a testing benchmark for dynamic algorithms (see discussion in \cite{peng2023fully}). (To the best of our knowledge, ours is the first use of this technique to control recourse, rather than space.)
Our results for sequences of interest imply that our algorithm of \Cref{thm-main}, which is competitive with such a solution (using $V_t$ a $(1-\eps)$-approximation of the current optimal solution), has polylogarithmic absolute recourse for such sequences, without being tailored to them. 

Without knowledge of $V_t$ (a natural variant of our problem), 
one may ask for an efficient approximation of $OPT$ at time $t$. Taking $V_t$ to be some polynomial-time computable $\alpha$-approximation of $OPT_t$. Depending on how precise the computed approximation $V_t$ is, the obtained approximation ratio is between $(1-1/e-\eps)$ and $\alpha(1-1/e-\eps)$, with recourse higher or lower, respectively.
Strikingly, if we ignore recourse, our approach yields a new meta-algorithm for $(1-1/e-\eps)$-approximate constrained submodular maximization algorithm:

\smallskip\noindent\textbf{New meta-algorithm for submodular maximization:} 
Key to our polynomial-time algorithms is a new meta-algorithm for maximizing the multilinear extension (\Cref{def:multilinear}) of a monotone submodular function under general constraints defined by a separable polytope~$\calP$.\footnote{Recall that polytope $\calP$ is \emph{separable} if there exists a polynomial-time \emph{separation oracle} that given $\vec{x}$ either determines that $\vec{x}\in \calP$, or outputs a separating hyperplane witnessing that $\vec{x}\not\in \calP$, i.e., a hyperplane separating $\vec{x}$~from~$\calP$.} 
We obtain a $(1-1/e-\eps)$-approximation, similar to the influential \emph{continuous greedy} algorithm~\cite{calinescu2011maximizing}.
Our algorithm uses as a black-box any underlying cutting plane method such as the Ellipsoid (e.g., \cite{Khachiyan80,Vaidya96,LSW15,jiang2020improved}) and requires only a separation oracle for $\calP$.
The key ingredient is a new procedure that we call {\em Approximate-or-Separate}, which informally either outputs a point $\vec{x}\in \calP$ which is a $(1-1/e-\eps)$-approximation, or else finds with good probability a separating hyperplane required for the underlying cutting plane algorithm. 
This generic approach for solving submodular maximization using any cutting plane algorithm as a ``black box" is similar in spirit to the exact unconstrained subomdular minimization, which was one of the first applications of cutting plane based algorithms in TCS \cite{grotschel1981ellipsoid}.

We anticipate further applications of this meta-algorithm and illustrate its use in \Cref{sec:applications}. We first provide static algorithms with optimal {\em curvature-sensitive} approximation, matching the results of \cite{sviridenko2017optimal,Feldman21}, but via simple black-box use of our meta-algorithm. We then develop communication complexity protocols, using our low-recourse submodular chasing objectives algorithm as the cutting plane method used.

We now overview this meta-algorithm, and other key ideas for submodular objective chasing.

\subsection{Technical Overview}\label{sec:techniques}

For simplicity, we focus here on 
submodular chasing with  cardinality constraint of $k$. However, our rounding extends to partition constraints, fractional chasing algorithms extend to matroids, and our new meta-algorithm extends to any separable polytope.

Our starting point is the fundamental {\em convex body chasing problem} \cite{friedman1993convex}: A sequence of convex shapes $\mathcal{K}_1,\mathcal{K}_2,\dots, \mathcal{K}_T \subseteq \mathbb{R}^n$ is revealed one at a time, and we must move a point so that it always lies inside the current convex shape, while minimizing overall movement according to some norm, $\sum_{t=1}^{T}\|\vec{x}^t-\vec{x}^{t-1}\|$.~For the $\ell_1$ norm, most relevant to recourse,  $\Omega(\sqrt{n})$ competitiveness is unavoidable \cite{bubeck2020chasing}, and $O(n)$ is attainable (for any norm) by the celebrated result of~\cite{sellke2020chasing}.
In recent~work \cite{bhattacharya2023chasing} showed that for convex shapes given by mixed packing-covering polytopes, by allowing a small violation of $\eps>0$ to the constraints, exponentially better competitive recourse is possible. 

\textbf{Slowly Chasing Submodular Objectives.}
We begin by relaxing the submodular objective chasing problem by ``dropping integrality constraints''.
That is, we allow both the optimal offline algorithm and ours to compute a sequence of points $\vec{x}^1,\dots,\vec{x}^T\in [0,1]^E$ satisfying $\sum_i x_i^t\leq k$ for all $t$.
This relaxation replaces recourse $\sum_t |S_t\oplus S_{t-1}|$ with total distance moved, $\sum_t \|\vec{x}^t-\vec{x}^{t-1}\|_1$. 
Finally, to extend the submodular set function from binary-valued to real-valued vectors, we use the Wolsey coverage extension, $f^*_t$ (\Cref{def:wolsey}), and add covering constraints equivalent to~$f^*_t(\vec{x}^t)\geq V_t$. This gives rise to a sequence of polytopes defined by packing-covering LPs, for which the chasing positive bodies framework of \cite{bhattacharya2023chasing} yields, using low fractional recourse compared to the optimal offline algorithm, feasible fractional vectors $\vec{x}^0,\vec{x}^1,\dots$ of high value, i.e., $f^*_t(\vec{x}^t) \geq (1-\eps) V_t$.
By standard connections (\Cref{lem:F>=(1-1/e)f*}), our fractional points therefore have high value according to independent rounding, a.k.a.~the multilinear extension, $F_t$ (\Cref{def:multilinear}); specifically, $F_t(\vec{x}^t)\geq (1-1/e-\eps)\cdot V_t$, which we later show how to attain integrally with the same recourse, via rounding.

\textbf{Polynomial-time algorithms.} Our relaxation is problematic from the point of view of fast algorithms, since $f^*(\vec{x})$ is APX-hard \cite{vondrak2007submodularity}. So, (unless \textsc{P}$=$\textsc{NP}) no efficient separation oracles exist  for $f^*_t(\vec{x}^t)\geq V_t$, as needed by the (cutting plane-based) algorithm of \cite{bhattacharya2023chasing}.
To overcome the computational intractability of $f^*$, inspired by and strengthening the \emph{round-or-separate} method of \cite{gupta2020online}, we provide an algorithmic counterpart of \Cref{lem:F>=(1-1/e)f*}: We provide an efficient \emph{approximate-or-separate}~method, which given point $\vec{x}$, monotone submodular function $f$ and values $V,\eps$, either (1) ascertains that $\vec{x}$ approximates $V$, in that $F(\vec{x}) \geq (1-1/e-\eps)V$, or (2) separates the constraint $f^*(\vec{x})\geq V$, finding a highly-violated covering constraint witnessing $f^*(\vec{x})\leq (1-\eps)V$.
Fixing such a constraint by moving to some $\vec{x}'$ satisfying $f^*_t(\vec{x}')\geq V_t$ requires non-trivial movement, which cannot happen often, by our recourse bounds.
This underlies our poly-time (fractional) submodular objectives chasing algorithm (\Cref{sec:chasing-optimally}). 
More generally, our new meta-algorithm uses the approximate-or-separate subroutine in any cutting plane method to separate the generally intractable constraint $f^*(\vec{x})\geq V$, or alternatively terminate early with a feasible vector $\vec{x}\in \calP$ providing a good approximation, $F(\vec{x})\geq (1-1/e-\eps)\cdot V$
(\Cref{sec:new-static}).

\textbf{Rounding.} 
Our fractional algorithm outputs a solution satisfying the slightly stronger guarantee, $F_t(\vec{1}- e^{-\vec{x}^t})\geq (1-1/e-\eps)\cdot V_t$, where $\vec{1}- e^{-\vec{x}}$ is a vector whose $i^{th}$ coordinate is $1-\exp(-x_i)$ for all $i\in E$. 
To obtain integral solutions of high value we use independent repetition of marginal-preserving and recourse-respecting algorithms for rounding a vector $\vec{x}$ satisfying a simple cardinality constraint, $\sum_i x_i\leq 1$ (e.g., \cite{keyfitz1951sampling}), using it to output a feasible vector with each $i$ sampled with probability $(1-(1-x_i^t/k)^k)\geq (1-\exp(-x_i^t))$.
The high submodular value follows from the output distribution also satisfying negative association \cite{joag1983negative}, and hence having expected value at least as high as independent rounding with the same marginals \cite{christofides2004connection,qiu2022submodular}, namely $F(\vec{1}-\exp(-\vec{x}))$.

\subsection{Further Related Work}\label{sec:related-work}

\paragraph{Consistent Submodular Maximization.}
For the special case of our problem with fixed $f$, i.e., $f_t=f$ for all times $t$, subject to cardinality constraints and incremental settings, $E_t = E_{t-1}\cup \{e_t\}$ for some $e_t\in E\setminus E_t$, many results were obtained in the context of preemptive online algorithms \cite{chan2018online,buchbinder2019online} and streaming algorithms \cite{chakrabarti2015submodular,chekuri2015streaming}. It has recently been studied with the explicit objective of minimizing \emph{worst-case} recourse \cite{dutting2024consistent,dutting2025cost}, with the restrictive requirement that $\underline{\max}_t \E[|S_t\oplus S_{t-1}|]$ be small.
\cite{dutting2025cost} showed that with constant worst-case recourse, the optimal approximation ratio is $2/3$ information-theoretically, and $0.51$ is achievable efficiently using randomization, while $0.5+\eps$ requires $\Omega_\eps(k)$ worst-case recourse deterministically \cite{dutting2024consistent}.
In contrast, the line of work on \emph{fast} dynamic constrained submodular maximization \cite{chen2022complexity,lattanzi2020fully,banihashem2024dynamic,banihashem2023dynamic} yields algorithms that are $(1-1/e-\epsilon)$-approximate under insertion-only streams and $(1/2-\epsilon)$-approximate in fully-dynamic settings with $O_{\eps} (\poly \log n)$ \emph{amortized} recourse, i.e., $\sum_{t=1}^T \E[|S_t\oplus S_{t-1}|]\leq T\cdot \poly \log n$ for fixed $\eps$.
None of these algorithms match our absolute recourse guarantees when the optimal offline recourse is a sufficiently low $o(T)$.

\section{Main Technical Background}\label{sec:prelims}

We study non-negative monotone submodular set functions $f\colon 2^E \to \R_+$ with $|E|=n$. A function $f$ is \emph{non-negative}  if $f(S) \geq 0$ for every set $S \subseteq E$, and \emph{monotone} if $f(S) \leq f(T)$ whenever $S \subseteq T \subseteq E$.
We use the shorthand $f(i \mid S) \triangleq f(S \cup \{i\}) - f(S)$ to denote the marginal contribution of element $i$ to set $S$. A function $f$ is \emph{submodular} if $f(i \mid S) \geq f(i \mid T)$ for every two sets $S \subseteq T \subseteq E$ and element $i \in E \setminus T$. We only assume \emph{value oracle access} to $f$, which given a set $S \subseteq E$ returns~$f(S)$.
A well-studied special case of (monotone) submodular functions are \emph{matroid rank functions},  $rank_\calM(\cdot)$, mapping a set $S$ to the size of a largest independent set of matroid $\calM$ in $S$. 
We will not need further properties of matroids, but we recall that these are ubiquitous in combinatorial optimization.

\paragraph{Extensions.} We use two extensions of submodular set functions to vectors in $[0,1]^{E}$. The first, corresponding to independent rounding, is the \emph{multilinear extension} of \cite{calinescu2011maximizing}. 
\begin{Def}\label{def:multilinear}
    The \emph{multilinear extension} $F:[0,1]^E\to \mathbb{R}$ of a set function $f:2^E\to \mathbb{R}$ is given by $$F(\vec{x})=\sum_{S\subseteq E} f(S) \prod_{i\in S} x_i \prod_{i\not\in S} (1-x_i).$$
\end{Def}
\begin{fact}\label{fact:ML-scaling}
    For any non-negative monotone submodular function $f:2^E\to \mathbb{R}_+$, real $p\in [0,1]$ and vectors $\vec{x}\in [0,1]^E$ and $\vec{y}\in [0,1]^E$ with $\vec{y}\geq p\cdot \vec{x}$,
    $$F(\vec{y}) \geq p\cdot F(\vec{x}).$$
\end{fact}
\begin{proof}
    By submodularity of $f$, for all non-negative $\vec{x}\geq \vec{0}$, the  function $F(\vec{0}+r\cdot \vec{x})$ is a concave function of $r\in [0,1]$. Therefore,  combined with monotonicity and non-negativity of $f$, we get
    \begin{align*}
    F(\vec{y}) & \geq F(p\cdot \vec{x}) = F((1-p)(\vec{0}+ 0 \cdot \vec{x})+ p ((\vec{0}+1 \cdot \vec{x})) \geq (1-p)F(\vec{0})+ p \cdot F(\vec{x})\geq p \cdot F(\vec{x}). \qedhere 
    \end{align*}
\end{proof}

The second extension we need is sometimes referred to as the \emph{Wolsey extension} \cite{levin2022submodular}, due to its roots in the work of Wolsey on submodular set coverage \cite{wolsey1982analysis}.  We use the more evocative name \emph{coverage extension} suggested by \cite{levin2022submodular}, given our use of this extension in coverage constraints.
\begin{Def}\label{def:wolsey}
    The \emph{coverage extension} $f^*:[0,1]^E\to \mathbb{R}$ of a set function $f:2^E\to \mathbb{R}$ is given by $$f^*(\vec{x})=\min_{S\subseteq E} \bigg(f(S) + \sum_{i\in E} f(i\mid S)\cdot x_i\bigg).$$
\end{Def}

For monotone submodular $f$, both $F$ and $f^*$ are extensions: $F(\mathds{1}_S) = f^*(\mathds{1}_S)=f(S)$~\cite{vondrak2007submodularity}.

The following shorthand will reoccur frequently in this paper.
\begin{Def}
    For any vector $\vec{x}\in [0,1]^E$, we denote by $\vec{1}-\exp(-\vec{x})$ the vector whose $i^{th}$ coordinate is $1-\exp(-x_i)$ for all $i\in E$.
\end{Def}

The following lemma was originally used by Vondr\'ak \cite{vondrak2007submodularity} to prove that the multilinear extension of any monotone submodular function $(1-1/e)$-approximates the submodular value-maximizing distribution with the given marginals, $f^+$, which is dominated by $f^*$ \cite{vondrak2007submodularity}.
While the \emph{statement} of \cite[Lemma 3.8]{vondrak2007submodularity} does not mention $F(\vec{1}-\exp(-\vec{x}))$, this stronger version which we state (and need) is precisely what Vondr\'ak proves (and we reprove later).

\begin{restatable}{lem}{vondrak}\label{lem:F>=(1-1/e)f*}
    Let $f:2^E\to \mathbb{R}_+$ be a monotone subdmodular function and $\vec{x}\in [0,1]^E$ be a vector. Then,
    $F(\vec{x}) \geq F(\vec{1}-\exp(-\vec{x}))\geq (1-1/e) \cdot f^*(\vec{x})$.
\end{restatable}

\paragraph{Chasing Positive Bodies.}
In the positive bodies chasing problem, the (online) input is a sequence of non-empty packing-covering polytopes of the form $\calK_t \triangleq \{\vec{x}\in \mathbb{R}^E_+ \mid C^t\vec{x}\geq \vec{1}, \; P^t\vec{x}\leq \vec{1}\}$, where both $P^t$ and $C^t$ are non-negative matrices. The following theorem is proved in \cite{bhattacharya2023chasing} (stated here with a violation of the covering constraint, which is more appropriate for our problem).

\begin{thm}\label{thm-fractional-chasing}
For any $\eps>0$, there exists a (deterministic) online  algorithm which given a sequence of non-empty packing-covering polytopes of the form $\calK_t$ above, outputs at each time $t$ an
$$\vec{x}^t \in \calK^{{1-\epsilon}}_t \triangleq \{\vec{x}\in \mathbb{R}^E_+ \mid C\vec{x}\geq \vec{1-\boldsymbol{\epsilon}}, \; P\vec{x}\leq \vec{1}\}.$$
The algorithm's recourse is $\sum_t \|\vec{x}^t-\vec{x}^{t-1}\|_1 = O(\eps^{-1} \log(d\cdot \eps^{-1})) \cdot OPT_R$, for 
 $d$ the maximal number of non-negative coefficients in any covering constraint and $OPT_R\triangleq \min_{\vec{x}^t\in \calK_t\; \forall t} \sum_t \|\vec{x}^t-\vec{x}^{t-1}\|_1$.
\end{thm}
\section{Approximate-or-Separate}\label{sec:polytime-separate}\label{sec:new-static}
\label{sec:approx-or-separate}

In this section we design a new meta-algorithm for approximating the multilinear extension $F$. Its guarantees are very similar to those of the influential Continuous Greedy Algorithm~\cite{calinescu2011maximizing}, and it can be used as an alternative approach to the latter. 
An advantage of this new meta-algorithm is its simplicity, and flexibility, allowing to implement it in various computational models, as we illustrate in subsequent sections.

Conceptually, the meta-algorithm is simple, and uses any cutting plane-based method, such as the Ellipsoid algorithm \cite{Khachiyan80,Vaidya96,LSW15,jiang2020improved}.
It is given (value oracle access to) a non-negative monotone  submodular function $f$ and a non-empty separable polytope $\calP\subseteq [0,1]^E$. It first guesses a value $V$ with $V/f(OPT)\in [1-\eps,1]$, for $f(OPT)\triangleq \max_{\vec{x}\in \{0,1\}^E\cap \calP}f(\vec{x})$. (This is standard and introduces only a logarithmic overhead; see the proof of \Cref{thm:guessing}.) Next, it uses any cutting plane-based method to find a point $\vec{x}$ in the polytope $\calP\cap \{\vec{x}\mid f^*(\vec{x})\geq V\}$, which is non-empty as $\calP$ is non-empty and $V\leq f(OPT)$. Such $\vec{x}$ is feasible (is in $\calP$) and satisfies by \Cref{lem:F>=(1-1/e)f*}  $F(\vec{x})\geq (1-1/e)\cdot f^*(\vec{x})\geq (1-1/e)\cdot V \geq (1-1/e-\eps)\cdot f(OPT)$.
Unfortunately, as outlined in the introduction, the constraint $f^*(\vec{x})\geq V$ is given by exponentially many constraints, and is moreover not separable, as even computing $f^*(\vec{x})$ is APX-hard \cite{vondrak2007submodularity}.

Our guiding intuition, inspired by \cite{gupta2020online}, is that we only wish to find a vector satisfying $f^*(\vec{x})\geq V$ as a means to compute a point $\vec{x}\in \calP$ with $F(\vec{x})\geq (1-1/e)\cdot V$ (or thereabouts). 
So, if we reach such a feasible point with high $F(\vec{x})$, we may terminate early, 
and we only require a separating hyperplane witnessing $f^*(\vec{x})<V$ when we are not at such a point.
The following simple lemma asserts that in the latter case a separating hyperplane \emph{exists}, and moreover this is equivalent to the classic \Cref{lem:F>=(1-1/e)f*}. Crucially for our rounding in \Cref{sec:rounding}, we even prove a slightly stronger claim that states that such a hyperplane exists even when $F(\vec{1}-\exp(-\vec{x}))<(1-1/e-\eps)\cdot V$.

\begin{lem}
    Let $f\colon 2^{E} \rightarrow \R_+$ be a non-negative monotone submodular function. Then, the following statements are equivalent.
    \begin{enumerate}
        \item \label{1:>=}\textbf{(\Cref{lem:F>=(1-1/e)f*})} For all $
        \vec{x}\in [0,1]^E$, we have $F(\vec{1}-\exp(-\vec{x}))\geq (1-1/e)\cdot f^*(\vec{x})$.
        \item \label{2:existential} \textbf{(High value or witness)} For all $\vec{x}\in [0,1]^E$, $\eps\in (0,1)$ and $V>0$, if $F(\vec{1}-\exp(-\vec{x}))<(1-1/e-\eps)\cdot V$, then there \underline{exists} a set $S\subseteq E$ satisfying $f(S)+\sum_i f(i \mid S)\cdot x_i\leq (1-\Theta(\eps))\cdot V$.
    \end{enumerate}
\end{lem}
\begin{proof}
    (\ref{1:>=} $\Rightarrow$ \ref{2:existential}) If $(1-1/e-\eps)V\geq F(\vec{1}-\exp(-\vec{x}))\geq (1-1/e)f^*(\vec{x})$ (the last inequality by \eqref{1:>=}), then we have that $(1-\eps/(1-1/e)) \cdot V \geq f^*(\vec{x}) = \min_{S\subseteq E} \left(f(S)+\sum_i f(i \mid S)\cdot x_i\right)$.

    (\ref{2:existential} $\Rightarrow$ \ref{1:>=}) If $F(\vec{1}-\exp(-\vec{x}))<(1-1/e-\eps)\cdot f^*(\vec{x})$ for some $\eps>0$, then by \eqref{2:existential} there exists a set $S\subseteq E$ with $f(S) + \sum_i f(i\mid S)\cdot x_i \leq (1-\Theta(\eps))\cdot f^*(\vec{x}) < f^*(\vec{x})$, contradicting the definition~of~$f^*(\vec{x})$.   
\end{proof}

Of course, the existence of a separating hyperplane does not imply efficient separation.
The core of our efficient algorithms is the following algorithmic counterpart to the existential statement \eqref{2:existential} above, using the exponential clocks of Vondr\'ak's proof of \Cref{lem:F>=(1-1/e)f*}. 

\begin{wrapper}
\begin{restatable}[Approximate-or-separate]{prop}{separateorapproximate}\label{prop-mainsep}
    Let $f\colon 2^{E} \rightarrow \R_+$ be a non-negative monotone submodular function, $\vec{x}\in [0,1]^n$, $\eps\in (0,1)$ and $V> 0$. If $F(\vec{1}-\exp(-\vec{x}))<(1-1/e-\eps) \cdot V$, then  
    the random set $S(t) \triangleq \{i\in E \mid  Y_i \leq t\}$ for independent $t\sim \Uni[0,1]$ and $Y_i\sim \Exp(x_i)$ for all $i\in E$ satisfies  with probability at least~$\eps^2/2$:   \begin{equation*}
        f(S(t))+ \sum_{i\in E}f(i\mid S(t)) \cdot x_i \leq \left(1-\eps/2 \right) \cdot V.
    \end{equation*}
\end{restatable}
\end{wrapper}

\begin{proof} 
Fix time $t\in [0,1]$. By memorylessness of exponentials, up to $O(dt^2)$ terms, for any set~$S\subseteq E$:
\begin{align*}
    \E[f(S(t+dt))-f(S(t)) \mid S(t)=S] & = \bigg[\sum_{i\in E}f(i \mid S) \cdot x_i \bigg]dt.
\end{align*}
So, by total expectation over the set $S(t)$, in the limit as $dt$ tends to zero,  $\phi(t)\triangleq \E[f(S(t))]$ satisfies:
\begin{align*}
    \phi(t) + \frac{d \phi(t)}{dt} 
& = \E\bigg[f(S(t))+\sum_{i\in E}f(i \mid S(t)) \cdot x_i \bigg].
\end{align*}

Now, if $\phi(t)+\phi'(t)\geq U$ for all $t\in [0,1]$, since $\phi(0)\geq 0$, then $\phi(t)$ is the solution of an ordinary differential inequality with minimal solution $\phi(t) \geq (1-\exp(-t))\cdot U$, and so $\phi(1)\geq (1-1/e)\cdot U$.
For example, by definition of $f^*$, this holds with $U=f^*(\vec{x})$.
Since each element $i$ belongs to $S(1)$ with probability $\Pr[Y_i\leq 1]=1-\exp(-x_i)$, we have that $\phi(1)=F(\vec{1}-\exp(-\vec{x}))$, this concludes Vondr\'ak's proof of \Cref{lem:F>=(1-1/e)f*}. 
Now, since by hypothesis $\phi(1)=F(\vec{1}-\exp(-\vec{x})) < (1-1/e)(1-\eps)\cdot V$, intuitively $\phi'(t)\geq U-\phi(t)$ for $U=(1-\eps)V$ should not hold for all, or even a high enough fraction of, $t$. 

To formalize this intuition, consider the ``log-gap'' function,\footnote{We thank Kalen Patton for suggesting this function, which streamlines the proof presented here.} $$\psi(t)\triangleq -\ln(V(1-\eps)-\phi(t)).$$ By the chain rule, 
$$\psi'(t)=\frac{\phi'(t)}{V(1-\eps)-\phi(t)}.$$
This derivative is non-negative, by monotonicity of $f$ and hence $\phi$, together with the hypothesis implying that $\phi(t)\leq \phi(1) = F(\vec{1}-\exp(-\vec{x})) < (1-1/e-\eps)\cdot V < V(1-\eps)$.
If moreover $\phi(t)+\phi'(t)\geq V(1-\eps)$, or equivalently $\phi'(t)\geq V(1-\eps)-\phi(t)(> 0)$, then $\psi'(t)\geq 1$. Thus, $$\psi'(t)\geq \mathds{1}[\phi(t)+\phi'(t)\geq V(1-\eps)].$$
But since $\phi(1) \leq (1-1/e-\eps)\cdot V$ by the hypothesis, and $\phi(0)\geq 0$ by non-negativity of $f$, we have, 
\begin{align*}\psi(1)-\psi(0) & \leq -\ln(V/e)+\ln(V(1-\eps)) = 1+\ln(1-\eps) \leq 1-\eps.
\end{align*}
Combining the preceding steps, we have that indeed $\Pr_t[\phi(t)+\phi'(t) < V(1-\eps)]\geq \eps$, since
\begin{align*}\Pr_{t}[\phi(t)+\phi'(t) \geq V(1-\eps)] 
& \leq \int_0^1 \psi'(t)\;dt =\psi(1)-\psi(0) \leq 1-\eps. \end{align*}
But by Markov's inequality, for any fixed $t\in[0,1]$  satisfying $\phi(t) + \phi'(t) \leq V(1-\eps)$,  
\begin{align*}
\Pr_{S\sim S(t)}\bigg[f(S)+\sum_{i\in E}f(i \mid S) \cdot x_i & \geq V\left(1-\eps/2\right)\bigg] \leq 
\frac{1-\eps}{1-\eps/2} \leq 1-\eps/2.\end{align*}
So, with probability at least $\eps\cdot \eps/2$, the set $S(t)$ for $t\sim \Uni[0,1]$ satisfies the desired inequality. 
\end{proof}

\begin{rem}
     A statement similar to  \Cref{prop-mainsep} was given by \cite[Lemma 4.1]{gupta2020online}. Unfortunately is too weak for our purposes, as it is not stated in general enough terms (for them, $V=f(E)$). Our algorithm is also faster, requiring  $O_\eps(n)$ oracle calls, as opposed to their $\tilde{O}_\eps(n^2)$. Most importantly, our proposition yields an optimal $(1-1/e)$ factor, as opposed to their $1/e$ factor.
\end{rem}

We are now ready to formally present our new meta-algorithm: \Cref{alg:Approximate-or-Separate} receives a non-negative monotone submodular function $f$, a non-empty separable polytope $\calP\subseteq [0,1]^E$ and values $V\in \R_+$ and $\eps,\delta\in (0,1)$. The algorithm guarantees that if $V\leq \max_{\vec{x}\in \calP}f^*(\vec{x})$, then with probability $1-\delta$ it outputs a vector $\vec{x}\in \calP$ such that $F(1-e^{\vec{x}})\geq (1-1/e-\eps)\cdot V$.

To this end, let $ALG$ be any cutting plane method, requiring at most $T$ cutting planes to solve $\calK(V)\triangleq \calP \cap \{\vec{x}\mid f^*(\vec{x})\geq V\}$.\footnote{For simplicity, the reader is recommended to consider a deterministic $ALG$, though one can use a randomized algorithm and union bound over its error probability and ours.}
That is, after querying at most $T$ vectors $\vec{x}$ for which we find some hyperplane separating $\vec{x}$ from $\calK(V)$ if $\vec{x}\notin \calK(V)$, $ALG$ either determines correctly that $\calK(V)=\emptyset$, or finds a point $\vec{x}\in \calK(V)$.
Our generic algorithm simply executes $ALG$ enough steps and uses \Cref{prop-mainsep} (with standard repetition argument to boost the success probability) to find a valid separating hyperplane for ALG whenever $F(\vec{1}-e^{\vec{x}}) <  (1-1/e-\eps)\cdot V$. The algorithm either terminates early if it finds $\vec{x}\in \calP$ and $F(\vec{1}-e^{\vec{x}})$ a good enough approximation (even if $\vec{x} \not \in \calK(V)$), or it finds another separating hyperplane for $\calK(V)$ (with good probability). Hence, similarly to \Cref{prop-mainsep}, we also name our \Cref{alg:Approximate-or-Separate} {\em Approximate-or-Separate}. 

\begin{algorithm}[t]
	\caption{Approximate-or-Separate $(f:2^E\to \R_+,\; \calP\subseteq [0,1]^E,\; \eps,\delta\in (0,1],\; V\in \R_+)$}\label{alg:Approximate-or-Separate}
 \begin{algorithmic}[1]
    \For{$T$ steps needed by a cutting plane-based method $ALG$ on $\calK(V)\triangleq \calP \cap \{\vec{x}\mid f^*(\vec{x}) \geq V\}$} 
    \If{$ALG$ queries $\vec{x}\in[0,1]^{E}$}
    \If{$\vec{x}\notin \calP$}
    \State Provide $ALG$ with a hyperplane separating $\vec{x}$ from $P$.
    \Else 
    \For{$O(\eps^{-2} \cdot \log(T \cdot \delta^{-1}))$ iterations}
    \State Let $S \triangleq \{i\in E \mid  Y_i \leq t\}$, for independent $t\sim \Uni[0,1]$ and $Y_i\sim \Exp(x_i)$ $\forall i\in E.$ \label{step:sample}
    \If{$f(S) + \sum_{i\in E} f(i \mid S)\cdot x_i < \left(1-\eps\right) \cdot V$}
\State Provide $ALG$ with the separating hyperplane
$f(S) + \sum_{i\in E} f(i \mid S)\cdot x_i \geq V$.
\EndIf
\EndFor
\EndIf
\If{no separating hyperplane found in this step of $ALG$}
\State \textbf{Return} $\vec{x}$.
\EndIf
\EndIf
\EndFor
\State \textbf{Return} $\calK(V)=\emptyset$. 
\end{algorithmic}	
\end{algorithm}

\begin{thm}\label{thm-fractional1}
Let $f\colon 2^{E} \rightarrow \R_+$ be a non-negative monotone submodular function
and $\calP\subseteq [0,1]^E$ be a separable polytope with $|E|=n$.
Then, whenever \Cref{alg:Approximate-or-Separate} outputs a vector $\vec{x}$, then $\vec{x}\in \calP$ (always) and $F(\vec{x})\geq F(\vec{1}-\exp(-\vec{x})) \geq  (1-1/e-\eps) \cdot V$ with probability $1-\delta$. Moreover, if $V\leq \max_{\vec{x}\in \calP}f^*(\vec{x})$ then the algorithm always outputs a vector $\vec{x}$.
\Cref{alg:Approximate-or-Separate} requires $O(\frac{nT}{\eps^2}\cdot \log (\delta^{-1}T))$ value oracle queries to $f$ and at most $T$  separation oracle queries to $\calP$, where $T$ is an upper bound on the number of separation queries required by $ALG$ on $\calK(V)$.
\end{thm}
\begin{proof}
If $V\leq \max_{\vec{x}\in \calP}f^*(\vec{x})$, then $\calK(V)=\calP \cap \{\vec{x} \mid f^*(\vec{x}) \geq V\}\neq \emptyset$. In this case, the algorithm would never report that $\calK(V)=\emptyset$.\footnote{Here we assume that $ALG$ is deterministic, but if it is a Monte Carlo randomized algorithm with error probability $\delta$, then by the union bound our error probability would at most double.} By \cref{prop-mainsep}, whenever the algorithm queries $\vec{x}\in \calP$ with $F(\vec{1}-\exp(-\vec{x}))<(1-1/e-\eps) \cdot V$, it will find a separating hyperplane with respect to the polytope $\calP \cap \{\vec{x} \mid f^*(\vec{x}) \geq V\}$ with probability $1-{\delta}/{T}$. This follows as the probability of success is $\Omega(\eps^2)$ and the algorithm makes $O(\eps^{-2} \cdot \log(T\delta^{-1}))$ independent calls to \cref{prop-mainsep}. As there are at most $T$ such steps, the probability that in any of these rounds the algorithm fails to provide a separating hyperplane is at most $\delta$, by the union bound. Hence, the algorithm must output after at most $T$ steps a vector $\vec{x}\in \calP$ such that $F(\vec{x})\geq F(\vec{1}-\exp(-\vec{x})) \geq (1-1/e-\eps) \cdot V$. Finally, checking the constraint for each random sample $S$ requires $n+1$ value oracle calls to $f$.
\end{proof}

By standard guessing of the optimal value $V$ up to a $(1-\eps)$ factor (see \Cref{sec:deferred-fractional}) combined with  \Cref{alg:Approximate-or-Separate}, we obtain our meta-algorithm for constrained submodular maximization.

\begin{restatable}{thm}{guessing}\label{thm-fractional2}\label{thm:guessing}
Let $f\colon 2^{E} \rightarrow \R_+$ be a normalized monotone submodular function, $\calP\subseteq [0,1]^E$ a separable polytope with $|E|=n$, and $ALG$ be a cutting plane method that requires $T$ or fewer separation queries. Then, there is an algorithm that for any $\eps>0$, outputs a vector $\vec{x}\in P$ satisfying $\E[F(\vec{x})]\geq (1-1/e-O(\eps))\cdot \max_{\vec{x}\in \calP}F(\vec{x})$.
The algorithm requires $\tilde{O}(T)$ separation oracle queries to $\calP$,
and  $\tilde{O}({n \cdot T}\cdot {\eps^{-2}})$ value oracle queries to $f$, where $\tilde{O}$ hides polylogarithmic terms in $n,T$ and~$1/\eps$.
\end{restatable}

In subsequent sections we illustrate this meta-algorithm's use for submodular maximization in various computational models, starting with polytime submodular objective chasing.

\section{Chasing Submodular Objectives Fractionally}\label{sec:chasing-optimally}

In this section we draw the connection between chasing submodular objectives and chasing positive bodies. We first define a fractional relaxation of the submodular function chasing problem. 
\begin{wrapper}\textbf{The Fractional Submodular Objectives Chasing Problem:}\\ 
	Initially, we are given a polytope $\calP \subseteq [0,1]^{E}$.
	Then, at each time $t=1, \ldots, T$ we are revealed a
	subset $E_t \subseteq E$ defining $\calP_t\triangleq \calP\cap \{\vec{x}\mid x_i=0 \textrm{ for all }i\not\in E^t\}$, a non-negative monotone submodular function $f_t:2^{E_t}\to~\mathbb{R}_+$,  
	and a desired value $V_t \leq OPT_t\triangleq \max_{S\subseteq E_t:\mathds{1}_S\in \calP_t}f_t(S)$. We must then output (before time $t+1$) a vector $\vec{x}^t\in \calP_t$ that is \emph{$\alpha$-approximate}, $\E[F(\vec{x}^t)]\geq \alpha\cdot V_t$, while using (ideally low) \emph{$c$-competitive} recourse, i.e., $\sum_{t=1}^{T}\|\vec{x}^t-\vec{x}^{t-1}\|_1\leq c\cdot OPT_R$, for $OPT_R$ the optimal (offline) recourse, 
	$$OPT_R\triangleq\min\left\{\sum_{t=1}^{T}\|\vec{x}^t-\vec{x}^{t-1}\|_1 \;\middle\vert\; \forall t: S_t\subseteq E_t:\mathds{1}_S\in \calP_t,\; f_t(S_t)\geq V_t \right\}.$$
\end{wrapper}

For the above problem with packing polytopes $\calP$, our non-polynomial \Cref{alg:chasing-slowly} is a simple reduction to positive body chasing (\cref{thm-fractional-chasing}), analyzed in part using \Cref{lem:F>=(1-1/e)f*}. It serves as a warm-up to our faster algorithm, below.\footnote{As an aside, for slowly chasing submodular objectives we propose the name \emph{trailing submodular objectives}.}

\begin{algorithm}[H]
	\caption{Chasing Submodular Objectives (Slowly)}\label{alg:chasing-slowly}
	\begin{algorithmic}[1]
		\State $\vec{x}^0\gets \vec{0}.$
		\For{time $t$}
		\State Reveal $\calK_t\triangleq \calP_t\cap \{\vec{x} \in \R_+ \mid f_t^*(\vec{x})\geq V_t\}$ to the algorithm of \Cref{thm-fractional-chasing} to get vector $\vec{x}^t$.
		\EndFor
	\end{algorithmic}	
\end{algorithm}	
\begin{thm}\label{lem-main-fractional}
	For any $\eps>0$ and packing polytope $\calP\subseteq [0,1]^E$,  \Cref{alg:chasing-slowly} is a deterministic $(1-1/e-\eps)$-approximate $O(\eps^{-1}\log(\max_t |E_t|\cdot \eps^{-1}))$-competitive algorithm for the fractional submodular objective chasing problem. Moreover, we have the stronger approximation guarantee that $F(\vec{1}-\exp(-\vec{x})) \geq (1-1/e-\eps)\cdot V_t$.     
\end{thm}
\begin{proof}
	Since $\calP$ is a packing polytope and since $f^*_t (\vec{x})=\min_{S\subseteq E_t}(f(S)+\sum_{i\in E_t}f(i\mid S)\cdot x_i)$, with $f_t(i\mid S)\geq 0$ by monotonicity of $f_t$, the $\calK_t$ of \Cref{alg:chasing-slowly} are packing-covering polytopes. Moreover, by our assumption that $V_t \leq \max_{S\subseteq E_t:\mathds{1}_S\in \calP_t}f_t(S)$, these $\calK_t$ are non-empty. 
	We can therefore provide them as input to the (deterministic) positive body chasing algorithm of \Cref{thm-fractional-chasing}, which returns at each time $t$ a vector $\vec{x}^t$ that satisfies for any $S\subseteq E_t$, 
	$$\sum_{i\in E_t}f(i \mid S)\cdot x^t_i \geq (1-\eps)\cdot (V_t-f(S)) \geq (1-\eps) \cdot V_t -f(S).$$ Hence, $f_t^*(\vec{x}^t)\geq (1-\eps) V_t$, and by \Cref{lem:F>=(1-1/e)f*}, we have $F_t(\vec{x}^t)\geq (1-1/e)(1-\eps)V_t\geq (1-1/e-\eps)V_t$.
	Finally, the covering constraints of $\calK_t$, of the form $f(S) + \sum_{i\in E_t} f(i\mid S)\cdot x_i \geq V_t$, have sparsity $\max_t |E_t|$, and so by \Cref{thm-fractional-chasing}, \Cref{alg:chasing-slowly}'s recourse is $O(\eps^{-1}\log (\max_t |E_t|\cdot \eps^{-1}))\cdot OPT_R$.\footnote{Indeed, $OPT_R$ above may even be replaced by the (possibly smaller) optimal \emph{fractional} algorithm's recourse.}
\end{proof}

The following theorem achieves the above guarantees \emph{efficiently}, by building on the Approximate-or-Separate framework, using that \cite{bhattacharya2023chasing} is essentially a cutting plane method.

\begin{thm}\label{lem-main-fractional-fast}
	There is a randomized $(1-1/e-O(\eps))$-approximate $O(\eps^{-1}\log(\max_t |E_t|\cdot \eps^{-1}))$-competitive \emph{polynomial-time} algorithm for the fractional submodular function chasing problem for any $\eps>0$ and matroid polytope $\calP\subseteq [0,1]^E$.
	Moreover, its outputs $\vec{x}_t$ satisfy the stronger approximation guarantee $F_t(\vec{1}-\exp(-\vec{x}^t)) \geq (1-1/e-O(\eps))\cdot V_t$.   
\end{thm}

\begin{proof}
	The proof of \Cref{lem-main-fractional} established that $\calK_t=\calP_t\cap \{\vec{x} \in \R_+ \mid f_t^*(\vec{x})\geq V_t\}$ are valid input polytopes for \cite{bhattacharya2023chasing}, and therefore so are polytopes defined by subsets of the constraints defining $\calK_t$. 
	Next, we show how to combine ideas from \Cref{sec:polytime-separate}
	to derive a randomized polynomial-time algorithm when $\calP=\calP({\cal M})$ is the matroid polytope of ${\cal M}$.
	To this end, we unbox \Cref{thm-fractional-chasing} (specifically, quantifying and improving arguments in \cite[Appendix A]{bhattacharya2023chasing}).
	
	We observe that the positive bodies chasing algorithm in \cite{bhattacharya2023chasing} requires only a separation oracle to $\calK^{1-\eps}_t$. 
	Hence, we may indeed directly use at any time $t$ the framework presented in \Cref{sec:polytime-separate}. However, there are two obstacles that we address in the following. First, the algorithm in \cite{bhattacharya2023chasing} is only guaranteed to output a solution that approximately satisfies the constraints, and second, there is no known polynomial-time bound on the number of separating hyperplanes required by that algorithm (and hence on its running time). We address both problems.
	
	\cite{bhattacharya2023chasing}
	reduce chasing positive bodies to chasing positive \emph{halfspaces}, i.e., chasing polytopes given by a single packing or covering constraint, $\langle \vec{a},\vec{x}\rangle\leq 1$ or $\langle \vec{a},\vec{x}\rangle\geq 1$ for some $\vec{a}\geq \vec{0}$, with the halfspaces used by the algorithm at time $t$ given by some constraints of $\calK_t$. 
	Thus, any solution for $\calK_1,\calK_2,\dots,\calK_T$ induces a solution for the halfspaces chased with the same recourse.
	In particular, for $r\leq n$ the rank of the matroid, this optimal (offline) recourse is at most $2rT$ (replacing all $r$ elements in every step). 
	For any violated constraint on a current point $\vec{x}$, they provide a polynomial-time solvable convex program (equivalent to a simple multiplicative update) to determine the next point $\vec{x}'$ satisfying the current constraint, while obtaining an $O_\eps(\log \max_t|E_t|)$-competitive recourse, and by the above using at most $O_\eps(rT\log\max_t|E_t|)$ recourse.
	A single update step in \cite{bhattacharya2023chasing} that computes $\vec{x}'$ such that $\langle \vec{a},\vec{x}'\rangle= 1$ requires finding a single dual variable that is used for the multiplicative update. 
	This can be done using a binary search and requires time that is linear in $n$ and logarithmic in the maximum and minimum non-zero coefficients $a_i$.

	To obtain a polynomial-time bound from this (as opposed to the finite bound stated in general in \cite{bhattacharya2023chasing}), we give polynomial-time algorithms to find  violated constraints such that moving to a point satisfying it incurs $\Omega(\eps)$ recourse, which by the recourse bound above can only happen $T_{ALG}\triangleq O_\eps((rT\log\max_t|E_t|)/\eps)$ times.

	\paragraph{Finding a (very) violated constraint in the Matroid Polytope:}
	We allow the algorithm momentarily to slightly violate the matroid constraints and only require for all $S\subseteq E_t$, $\sum_{i\in S} x^t_i \leq (1+\eps)\cdot rank(S)$ (and fix this later on in the proof). Next, by solving unconstrained submodular minimization \cite{LSW15,jiang2020improved} of the function $g(S) \triangleq (1+\eps)\cdot rank(S) - \sum_{i\in S} x_i$, with our current vector $\vec{x}$, we either decide that all the relaxed matroid constraints are satisfied, or find a subset $S$, such that $\sum_{i\in S} x_i \geq (1+\eps)\cdot rank(S)$.
	In the latter case, moving to any point $\vec{x}'$ that satisfies the violated constraint $\sum_{i\in S} x_i \leq rank(S)$ incurs recourse $\|\vec{x}-\vec{x}'\|_1 \geq \eps \cdot rank(S) \geq \eps$.    
	\paragraph{Finding a violated constraint in the Wolsey Polytope:}
	Using the framework of \Cref{sec:polytime-separate} we only find such a violated constraint whenever $F_t(\vec{1}-\exp(-\vec{x})) < (1-1/e-\eps)\cdot V_t$, where $\vec{x}$ is our current vector, as otherwise we may output $\vec{x}$. In this case, we use  \Cref{prop-mainsep} to sample 
	$s\triangleq\lceil \eps^{-2}\log (T_{ALG}/\eps))\rceil$ sets, for $T_{ALG}=\Theta_\eps(rt\log\max_t|E_t|)$. As each sampled set $S$ satisfies $f(S) + \sum_{i\in E_t} f(i\mid S)\cdot x_i \leq \left(1-\eps/2\right)\cdot V_t.$ with probability $\Omega(\eps^2)$, we are guaranteed that one of the  sets $S$ satisfies such a violated constraint with probability $1-(1-\eps^{2})^{s}\geq 1-\eps/T_{ALG}$.
	Taking union bound over the $T_{ALG}$ steps of the algorithm in this time-step until we output some vector $\vec{x}$, we find that the probability we fail to witness that $F_t(\vec{1}-\exp(-\vec{x}))<(1-1/e-\eps)\cdot V_t$ at some point of our algorithm where this holds is at most $\eps$, and so $\E[F_t(\vec{1}-\exp(-\vec{x}))]\geq (1-\eps)\cdot (1-1/e-\eps)
	\cdot V_t\geq (1-1/e-2\eps)\cdot V_t$.
	This (together with subsequent scaling, discussed below) implies the desired approximation.
	
	\paragraph{Lower bounding recourse for fixing a violated constraint.} To fix the violated constraint by moving to a point~$\vec{x}'$ satisfying $f(S) + \sum_{i\in E_t} f(i\mid S)\cdot x'_i \geq V_t$ requires recourse 
	$\|\vec{x}-\vec{x}'\|_1\geq \left((\eps/2)\cdot V_t\right)/\max_{i\in E_t} f(i\mid S)$. Unfortunately, this can be very small, as $V_t$ may be small compared to $\max_{i\in E_t} f(i\mid S)$, which can be as large as $OPT_t$. However, we note that every {\bf integral} solution $\vec{x}\in \{0,1\}^E$ that satisfies the Wolsey constraint $f(S) + \sum_{i\in E_t} f(i\mid S)\cdot x_i \geq V_t$ also satisfies the stronger constraint 
	$f(S) + \sum_{i\in E_t} \min\{V_t- f(S), f(i\mid S)\}\cdot x_i \geq V_t$. Hence, as we are competing with the optimal integral recourse, this is a valid constraint for $\calK_t\cap \{\vec{x}\in\{0,1\}^{E_t}\}$, and we may use it as a separating hyperplane.
	Specifically, if we find a set $S\subseteq E_t$ violating the easier constraint, we also find that the stronger constraint is violated, and we may move to a point $\vec{x}'$ satisfying the latter.
	The recourse is now at least 
	\begin{align*}
		\|\vec{x}-\vec{x}'\|_1\geq \left((\eps/2)\cdot V_t\right)/\min\left\{V_t- f(S),\; \max_{i\in E_t} f(i\mid S)\right\} \geq \left((\eps/2)\cdot V_t\right)/V_t= \Omega(\eps).
	\end{align*}
	
	\paragraph{Conclusion.} Finally, we use the above two algorithms at each time $t$ to find halfspaces to provide to the algorithm of \Cref{thm-fractional-chasing} until we find no violated packing or covering constraints. The algorithm is guaranteed to output at any time a point $\vec{y}^t$ satisfying $\sum_{i\in S}y^t_i\leq (1+\eps) \cdot rank(S)$ for each set $S\subseteq E_t$ and with probability $1-\eps/T_{ALG}$ also 
	$F_t(\vec{1}-e^{-\vec{y}^t})\geq (1-1/e-2\eps)\cdot V_t$, and so also $\E[F_t(\vec{1}-e^{-\vec{y}^t})]\geq (1-1/e-O(\eps))\cdot V_t$.
	Finally, taking $\vec{x}^t\gets \vec{y}^t/(1+\eps)$ for all times $t$ satisfies the matroid polytope's constraints, does not increase the $O(\eps^{-1}\log(\max_t |E_t|\cdot \eps^{-1}))$-competitive recourse  of \Cref{thm-fractional-chasing}, and guarantees by monotonicity and \Cref{fact:ML-scaling} that
	\begin{align*}
		F(\vec{1}-e^{-\vec{x}^t})= F\big(\vec{1}-e^{-\frac{\vec{y}^t}{1+\eps}}\big) & \geq F\bigg(\frac{\vec{1}-e^{-\vec{y}^t}}{1+\eps}\bigg) \geq \frac{1}{1+\eps}F(\vec{1}-e^{-\vec{y}^t})\geq (1-1/e-O(\eps))\cdot V_t.\qedhere 
	\end{align*}
\end{proof}
\section{Recourse-Respecting Randomized Rounding}\label{sec:rounding}

In this section we provide a randomized rounding algorithm for submodular objective chasing. 

Our rounding algorithm converts online fractional solutions $\vec{x}^1,\dots,\vec{x}^T$ into randomized integral feasible sets $S_1,\dots,S_T$, of expected submodular value at least $F(\vec{1}-\exp(-\vec{x}^t))$, while respecting the recourse guarantee of the fractional solution; i.e., obtaining no higher (expected) recourse. The guarantees we desire are as follows.

\begin{Def}\label{def:rrrr}
    A cardinality-constrained recourse-respecting randomized rounding receives as inputs vectors $\vec{x}^1,\dots,\vec{x}^T\in [0,1]^E$ online, satisfying $\sum_i x^t_i\leq k$ for all $t$; such an algorithm outputs after each $\vec{x}^t$ is revealed (and before $\vec{x}^{t+1}$ is revealed) a set $S_t\subseteq E_t$ satisfying for all $t:$
\begin{enumerate}[label=(P{{\arabic*}})]
    \item \label{prop:k-uniform} \textbf{(Cardinality)} 
    $|S_t| \leq k$.
    \item \label{prop:distance}\textbf{(Recourse)} $\E[\;| S_t\oplus S_{t-1}|\;]\leq \norm{\vec{x}^t-\vec{x}^{t-1}}_1$.
    \item \textbf{(Marginals)}\label{prop:marginals} $x^t_i\geq \Pr[i\in S_t]\geq 1-(1-x^t_i/k)^k \geq 1-\exp(-x^t_i)$ for all $i\in E$.
    \item \textbf{(NA)} \label{prop:na} The characteristic vector of $S_t$, i.e., $\mathds{1}[S_t]\triangleq (\mathds{1}[i\in S_t] \mid i\in E)$, is negatively~associated.
    \item \label{prop:approx-submod-dom} \textbf{(Value)}
    $\E[f_t(S_t)]\geq F_t\left(\vec{1}-\exp\left(-\vec{x}^t\right)\right)$ for any monotone submodular function $f_t$.
\end{enumerate}
\end{Def}

Note that Property \ref{prop:marginals} implies that the output only contains elements with non-zero $x$-value.

The special case of $k=1$ has been studied in the statistics literature over 70 years ago. Specifically, Keyfitz provided a method with the following guarantees \cite{keyfitz1951sampling}.

\begin{prop}
    There exists a cardinality-constraint recourse-respecting randomized rounding algorithm for $k=1$.
\end{prop}

As Keyfitz's algorithm is mainly provided by way of example, for completeness we formalize this algorithm and its analysis in \Cref{app:unit-capacity}. In the same appendix we provide an alternative algorithm for the $k=1$ case requiring only a single $\Uni[0,1]$ random variable, as opposed to drawing such a random variable for each vector, as with Keyftiz's algorithm.

We now show how the $k=1$ case implies our desired guarantees for $k>1$ in general, by independent repetition. This will require the following background on negative association and its role in submodular maximization.

\paragraph{Submodular Dominance and Negative Association.} We require rounding satisfying the constraints and obtaining value at least as high as the multilinear extension (at the point $\vec{y}=\vec{1}-\exp(-\vec{x})$). We will obtain via randomized rounding satisfying \emph{negative association}.

\begin{Def}
    A random vector $\vec{X}$ is \emph{negatively associated (NA)} if for any two monotone functions $f,g$ of disjoint variables in $\vec{X}$, we have $\Cov(f(\vec{X}),g(\vec{X}))\leq 0.$
\end{Def}

As shown by \cite{christofides2004connection,qiu2022submodular}, negative association \cite{joag1983negative} implies \emph{subdmodular dominance}.

\begin{lem}\label{lem:NA->submod}
    If $\vec{X}\in \{0,1\}$ is NA and $f:2^E\to \mathbb{R}$ is submodular, then
    $\E[f(\vec{X})]\geq F(\E[\vec{X}])$.
\end{lem}
In our analysis we will need the following well-known properties of NA vectors \cite{joag1983negative,dubhashi1996balls}.
\begin{lem}\label{lem:na}
The following give rise to NA distributions:
    \begin{enumerate}
        \item \label{lem:0-1-lemma} \textbf{(0-1 rule)} Any random vector $\vec{X}\in \{0,1\}^n$ with $\sum_i X_i\leq 1$ is NA.
        \item \textbf{(Products)} The concatenation $\vec X \circ \vec Y$ of independent NA vectors $\vec X$ and $\vec Y$ is also NA.
        \item \textbf{(Composition)} If $\vec X$ is an NA vector and $f_1,\dots, f_r$ are monotone functions of disjiont variables in $\vec X$, then the vector $(f_1(\vec{X}),\dots,f_r(\vec{X}))$ is also NA.
    \end{enumerate}
\end{lem}

With the above in place, we are now ready to lift any recourse-respecting randomized rounding from the (trivial) $k=1$ case to arbitrary $k\geq 1$.

\begin{lem}
    Let $\calA$ be a cardinality-constrained recourse-respecting randomized rounding algorithm for the $k=1$ case. 
    For each time $t$, let $S_t^{(1)},\dots,S_t^{(k)}$ be the outputs of $k>1$ independent copies of $\calA$ on input stream $\vec{x}^1/k,\dots,\vec{x}^t/k$. Then an algorithm $\calA'$ outputting $S_t\triangleq \bigcup_{i=1}^k S_t^{(i)}$ is a cardinality-constrained recourse-respecting randomized rounding algorithm for $k>1$.
\end{lem}
\begin{proof}
    Properties \ref{prop:k-uniform} and \ref{prop:distance} follow by the same properties holding for $\calA$ and triangle inequality:  
    \begin{align*}
        |S_t| & =\left|\bigcup_{i=1}^k S_t^{(i)}\right|\leq \sum_{i=1}^k |S_t^{(i)}|\leq k.\\
        \E[\;|S_t \oplus S_{t-1}|\;] & \leq \E\left[\;\left|\bigcup_{i=1}^k S_t^{(i)} \oplus \bigcup_{i=1}^kS_{t-1}^{(i)}\right|\;\right] \leq \sum_{i=1}^k \E\left[\left| S^{(i)}_t\oplus S^{(i)}_{t-1}\right|\right] \leq \sum_{i=1}^k \norm{\vec{x}^t/k - \vec{x}^{t-1}/k} = \norm{\vec{x}^t - \vec{x}^{t-1}}.
    \end{align*}
    Similarly, Property \ref{prop:marginals} follows from the same property of the independent copies of $\calA$ and union bound:
    \begin{align*}
        \Pr[e\in S_t] & = \Pr\left[\exists i:\; e\in S_t^{(i)}\right] \in \left[1-(1-x^t_i/k)^k,\;\; x^t_i\right]. 
    \end{align*}
    Next, Property \ref{prop:na} follows from the same property of $\calA$ and the variables $\mathds{1}[e\in S_t] = \bigvee_i \mathds{1}[e\in S_t^{(i)}]$ being monotone increasing functions of disjoint NA variables.
    Finally, Property \ref{prop:approx-submod-dom} follows from the two preceding properties together with \Cref{lem:NA->submod} and monotonicity of $f$.
\end{proof}

\paragraph{Extension to Partition Matroids.}
As one might expect, applying a recourse-respecting randomized rounding algorithm independently with the appropriate cardinality constraint for each of the different parts of a partition matroid implies similar guarantees.
Here, for partition $P_1,\dots,P_r$ of $E$ we use the shorthand $\vec{x}_{\mid r}$ to indicate the projection of $\vec{x}$ onto the $|P_i|$-dimensional space indexed by elements in $P_i$.

\begin{thm}\label{thm:rounding-partition}
    Fix a partition constraint $((P_1,k_1),\dots,(P_r,k_r))$.
    Applying a recourse-respecting randomized rounding algorithm independently for each $j\in [r]$ to the vectors $\vec{x}^1_{\mid j},\dots,\vec{x}^T_{\mid j} \in [0,1]^E$ as they are revealed, for $\vec{x}^1,\dots,\vec{x}^T\in [0,1]^E$ (fractionally) satisfying the cardinality constraints, results in sets $S^t_{\mid j}$ whose union, $S_t\triangleq \bigcup_j S^t_{\mid j}$, is contained in $\{i \mid x^t_i\neq 0\}$ and satisfies the following:
\begin{enumerate}[label=(Q{{\arabic*}})]
    \item \label{partition:k-uniform} \textbf{(Partition)} 
    $S_t$ satisfies the partition constraint: $S_t\cap P_j = S^t_{\mid j}$ has cardinality at most $k_j$.
    \item \label{partition:distance}\textbf{(Recourse)} $\E[\;| S_t\oplus S_{t-1}|\;]\leq \norm{\vec{x}^t-\vec{x}^{t-1}}_1$.
    \item \label{partition:approx-submod-dom} \textbf{(Value)}
    $\E[f_t(S_t)]\geq F_t\left(\vec{1}-\exp\left(-\vec{x}^t\right)\right)$ for any monotone submodular function $f$.
\end{enumerate}
\end{thm}
\begin{proof}
    That $S_t\subseteq \{i\mid x^t_i\neq 0\}$ follows from the upper bound of Property \ref{prop:marginals}.
    Properties \ref{partition:k-uniform} and \ref{partition:distance} follow directly from \ref{prop:k-uniform} and \ref{prop:distance} (the latter together with linearity of expectation). To prove Property \ref{partition:approx-submod-dom}, 
    we rely on closure of NA under products and disjoint functions (\Cref{lem:na}), implying together with Property \ref{prop:na} that the characteristic vector of $S_t$ is NA.
    Property \ref{partition:approx-submod-dom} then follows by Property \ref{prop:marginals} and submodular dominance of NA vectors (\Cref{lem:NA->submod}) and monotonicity of $f$.
\end{proof}

Finally, combining our fractional submodular objective chasing algorithms and the preceding rounding schemes, we are ready to prove our main result.

\mainthm*
\begin{proof}
    By \Cref{{lem-main-fractional-fast}}, there exists a $(1-1/e-\eps)$-approximate and $O(\eps^{-1}\log(\max_t|E_t|\eps^{-1})$-competitive fractional submodular objective chasing algorithm for packing constraints, which include cardinality and partition constraints.
    Apply this algorithm to the input sequence $\langle f_t,E_t,V_t \rangle_t$ to obtain fractional solutions $\vec{x}^1,\dots,\vec{x}^T$.
    Whenever one such vector is revealed, apply the rounding algorithm of \Cref{thm:rounding-partition}, to obtain online sets $S_1,\dots,S_T$ that are: (1) feasible (by Property \ref{partition:k-uniform}), (2) have recourse no higher than the fractional solution, $\sum_t \E[\;|S_t\bigoplus S_{t-1} |\;] \leq \sum_t \E[\|\vec{x}^t-\vec{x}^{t-1}\|_1]$ (by Property \ref{partition:distance}), and  therefore this new algorithm is $O(\eps^{-1}\log(\max_t |E_t|\eps^{-1})$-competitive, and (3) the sets have expected value at least $
    \E[f_t(S_t)]\geq F_t(\vec{1}-\exp(-\vec{x}^t) \geq (1-1/e-\eps)\cdot V_t$, by Property \ref{partition:approx-submod-dom} and $\E[\mathds{1}_{S_t}]\geq \vec{1}-\exp(-\vec{x})$ (placewise), \Cref{lem:NA->submod}, and the stronger approximation guarantee of \Cref{{lem-main-fractional-fast}}. That is, the obtained integral algorithm is $(1-1/e-\eps)$-approximate.
\end{proof}

\begin{rem}[Concentration]
    By \cite{duppala2025concentration}, non-negative monotone submodular functions of sets with NA indicator vector satisfy lower tail bounds. Consequently, our output sets $S_t$ satisfy, for $m$ the maximum marginal at time $t$ and $\mu\triangleq F_t(\vec{1}-\vec{x}^t)$, 
    $$\Pr[f_t(S_t) < (1-\delta)\mu] \leq \exp\left(-\frac{\delta^{2}\mu}{2m}\right).$$
\end{rem}

\section{More Applications of Approximate-or-Separate}\label{sec:applications}

In this section we illustrate the flexibility of the approximate-or-separate proposition and algorithmic paradigm, by providing (1) a simple polynomial-time algorithm for constrained maximization of monotone submodular functions with bounded curvature, and (2) a communication-optimal communication complexity protocol for constrained submodular maximization. 


\subsection{Curvature-Sensitive Approximation}\label{sec:curvature}

The $(1-1/e)$-approximation for constrained submodular maximization is optimal in the worst case for general monotone submodular functions. However, it can be refined if the submodular function has bounded \emph{total curvature}, defined as follows \cite{conforti1984submodular}.
\begin{Def}\label{def:curvature}
A set function $f$ with no zero-valued elements has \emph{total curvature} $$c_f \triangleq 1-\min_{i\in E} \frac{f(i\mid E\setminus\{i\})}{f(\{i\})}.$$
\end{Def}

Conforti and Cournejols \cite{conforti1984submodular} showed that the greedy algorithm is $(1-\exp(-c))/c$-approximate for monotone submodular function maximization subject to cardinality constraints and $1/(1+c)$-approximate for any matroid constraints.
In an influential paper, \cite{sviridenko2017optimal} improved on both bounds, and provided two methods to provide a $(1-c/e-\eps)$-approximation for any matroid constraints, which they prove is optimal (up to the $\eps$ term) for all values of $c\in [0,1]$. Feldman \cite{Feldman21} later provided an alternative method to obtain this result.

The following decomposition of \cite{iyer2013curvature} is useful (see \Cref{appendix:curvature} for a self-contained proof):
\begin{restatable}{lem}{curvaturedecomposition}\label{lem:curvature-decomposition}
    Any non-negative monotone subdmodular function $f$ with curvature $c > 0$ can be written as $f=c\cdot g + (1-c)\cdot \ell$, where $\ell$ is linear, $\ell(S)\triangleq \sum_{i\in S} f(\{i\})$, and the resulting function $g=\frac{f-(1-c)\ell}{c}$ is a non-negative monotone submodular function satisfying $g\leq f$.
\end{restatable}

The improved curvature result of \cite{sviridenko2017optimal} follows directly from the next theorem, which applies their \cite[Theorem 3.1]{sviridenko2017optimal} to the functions $g$ and $\ell$ in the decomposition of $f$ guaranteed in \Cref{lem:curvature-decomposition}.


\begin{thm}
Let $\hat{v}\triangleq \max_{i\in E}\max\{g(\{i\}), \ell(\{i\})\}$.
For every $\eps>0$, there is an algorithm that given a monotone submodular function $g$, a linear function $\ell$ and a matroid $\calM$, produces in polynomial time a feasible subset $S \subseteq E$  satisfying:
\begin{align*}
\E[f(S)] & =\E[c\cdot g(S)+(1-c)\cdot \ell(S)] \\
& \geq (1-1/e)\cdot c\cdot g(O)+ (1-c)\cdot \ell(O) - O(\eps) \cdot \hat{v} \\
& \geq (1-c/e-O(\eps)) \cdot f(O),\end{align*}
where $g$ and $\ell$ are as in \Cref{lem:curvature-decomposition}, $\hat{v}\leq f(O)$, 
and $O\subseteq E$ is the feasible subset maximizing $f$.
\end{thm}





Two proofs for this theorem are provided by \cite{sviridenko2017optimal}. The first one works by a modification of the continuous greedy algorithm that yields a feasible vector $\vec{x}$ such that $G(\vec{x})+ L(\vec{x}) \geq (1-1/e)g(O)+ \ell(O) - O(\eps) \cdot \hat{v}$, where $G$ and $L$ are the multilinear extensions of $g$ and $\ell$, respectively. Then, this fractional vector $\vec{x}$ can be rounded to an integral solution without an additional loss \cite{calinescu2011maximizing}. The second proof relies on a modification of the local search algorithm of \cite{filmus2014monotone}.

We show that a very simple application (rather than modification) of \Cref{alg:Approximate-or-Separate} yields the same guarantee as \cite{sviridenko2017optimal}. The idea is simple: similarly to \cite{sviridenko2017optimal}, we guess two values $\lambda= \ell(O)$ and $\gamma=g(O)$. Then, we execute \Cref{alg:Approximate-or-Separate} on the polytope $\calK'(V)\triangleq \calP \cap \{\vec{x}\mid f^*(\vec{x}) \geq \gamma\} \cap \{\vec{x}\mid L(\vec{x}) \geq \lambda\}$, noting that if the values are guessed correctly, $\calK'(V)\neq \emptyset$. We further observe that we are only adding a single constraint ($L(\vec{x}) \geq \lambda$) and checking whether this constraint is violated is easy.
Therefore, by \Cref{thm-fractional1}, the final vector satisfies $G(\vec{x})+ L(\vec{x}) \geq (1-1/e)\cdot \gamma+\lambda$. The additional (additive) error of $\eps \cdot \hat{v}$ is due to the inaccurate guessing of the values $\gamma$ and $\lambda$. The procedure for guessing these values can be handled in the same way as it is done in \cite{sviridenko2017optimal}. Finally, rounding again yields the claimed bound in expectation.

Summarizing, at the cost of a small $\poly(n/\eps)$ slowdown due to guesswork (see \cite{sviridenko2017optimal} for details), the Approximate-or-Separate framework allows us to get curvature-sensitive algorithms with effectively the same running time as our $(1-1/e-\eps)$ approximation. This demonstrates the flexibility of our framework that allows adding more constraints (linear or even convex) easily, while relying only on a separation (or approximate separation) oracle required for the underlying cutting based algorithm.
\subsection{Communication Complexity}\label{sec:communication-complexity}

In this section we outline the utility of our approximate-or-separate framework to the communication complexity model \cite{kushilevitz2006communication}. 
In the basic model, Alice and Bob have inputs $a\in \calA$ and $b\in \calB$, respectively, unknown to each other, and must communicate over discrete back-and-forth rounds, ideally sending as little bits of information as possible in order to agree on some $f(a,b)$, where $f:\calA\times \calB\to \calC$ is known to both parties.
For example, Alice and Bob may hold parts of the edges of a graph, and may want to compute a perfect matching in the union graph, if one exists \cite{blikstad2022nearly}.

In our context, Alice and Bob have some subsets $A\subseteq [n]$ and $B\subseteq [n]$ of a common ground set~$[n]$.
They both know a common non-negative monotone submodular function $f:2^{[n]}\to \mathbb{R}_+$, a partition constraint of rank $r$ (i.e., maximum solution cardinality of $r$), and wish to compute a feasible subset $C\subseteq A\cup B$ that $(1-1/e-\eps)$-approximates the optimal feasible set $OPT\subseteq A\cup B$ with probability at least $1-\eps$.
Alice and Bob have their own random bits, which they may use to communicate. This is the so-called ($\eps$-error) \emph{private-coin} communication model, as opposed to the \emph{public-coin} communication model, where Alice and Bob also share a common random string.
By Newman's theorem \cite{newman1991private}, by doubling the error probability from $\eps$ to $2\eps$, the communication complexity in both models differs by at most an additive $O_\eps(\log n)$ bits. 

The problem was previously studied by Feldman et al.~\cite{feldman2023one} for cardinality constraint. They showed that even in the one-way communication model (Alice only send a single message to Bob), non-trivial guarantees can be achieved: a $0.51$-approximation can be achieved in polytime by sending a message of $O(r)$ elements, while a $(2/3-\eps)$-approximation can be achieved (in exponential time) by sendig a message of $O(r\log(r/\eps))$ elements. In this section we will show that $\tilde{O}(r)$ many elements, or $\tilde{O}(r)$ bits, suffice to obtain a $(1-1/e-\eps)$-approximation for partition constraints more generally, using back and forth communication.

Since the solution can be of cardinality $r$, it is intuitively clear that $\Omega(r)$ bits are necessary for Alice and Bob to communicate an approximately-maximum solution $C\subset A\cup B$. The next fact, whose proof is a simple counting argument, noting that there are $\sum_{k=\alpha r}^r {n \choose r} = \Omega({n\choose \alpha r})$ possible outputs Alice and Bob must agree on to get an $\alpha$-approximate solution depending only on Alice's input if Bob's input is empty.

\begin{restatable}{fact}{cclb}\label{fact:cclb}
	The private-coin communication complexity of any constant-approximate cardinality-$r$-constrained non-negative submodular function minimization is $\Omega(r \log(n/r))$.
\end{restatable}

We now show that a nearly-matching $r\cdot \poly(1/\eps,\log n)$ bits is achievable by a randomized protocol relying on the approximate-or-separate paradigm, instantiated using the low-round and low-recourse algorithm used for submodular objective chasing in \Cref{sec:chasing-optimally}.

Before doing so, we make some simple simplifying assumptions.

\begin{obs}\label{obs:guess-opt}
	Using $O(\log n)$ bits, Alice and Bob can communicate $\max_{i\in A\cup B} f(i)$.
\end{obs}
\begin{proof}
	This can be done by letting Alice and Bob communicating some $i\in \arg\max_{i\in A} f(i)$ and $i\in \arg\max_{i\in B} f(i)$ to each other, and then evaluate both singletons using the oracle access to $f$.
\end{proof}
By standard sub-additivity arguments, implying that ignoring any element of marginal value at most $(\eps/n)\cdot \max_{i\in A\cup B}$ cannot affect the value of any set by more than $\eps\cdot \max_{i\in A\cup B} f(i) \leq \eps \max_{C\subseteq A\cup B} f(C)$, we can therefore assume that transmitting any marginal requires $O(\log (n/\eps))$ bits.

We now proceed to (approximately) guessing $f(OPT)$, needed by our approximate-or-separate method.

\begin{obs}[Guessing $f(OPT)$]
	Using $O(r\log n)$ communication, and an $O(\eps^{-1})$ multiplicative overhead, Alice and Bob can be assumed to know some value $V$ such that $V/f(OPT)\in [1-\eps,1]$.
\end{obs}
\begin{proof} Alice and Bob can agree on some $\Omega(1)$-approximate solution value, by exchanging $\Omega(1)$-approximate optimal solutions in $A$ and in $B$, using $O(r\log n)$ bits. 
	This way, they can use $\log_{1/(1-\eps)}(\Theta(1))=O(1/\eps)$ many guesses for $f(OPT)$, one such value $V$ of which $(1-\eps)$-approximates this value. 
	That is, $V/f(OPT)\in [1-\eps,1]$. Running the protocol with $O(1/\eps)$ guesses for $f(OPT)$ will then allow them to agree on the best solution for all these guesses.
\end{proof}

We now describe the algorithm for each of the above guessed values $V$, where the correct guess will result with probability at least $1-\eps$ in an output solution of value $(1-1/e-O(\eps))\cdot f(OPT)$ in expectation. By reverse Markov, since their output solution has value in $[0,f(OPT)]$ by definition, the probability that the output has value less than $(1-1/e-O(\eps))\cdot f(OPT)$ is $1-O(\eps)$. Therefore, repeating this algorithm $\tilde{O}(1/\eps)$ many times and taking the best solution therefore results in success with high probability. We therefore focus on proving the desired approximation holds in expectation for a guessed value $V$ with $V/f(OPT)\in [1-\eps,1]$.

\begin{thm}
	The private-coin randomized communication complexity of $(1-1/e-\eps)$-approximate partition-constrained non-negative monotone submodular function maximization is $r\cdot \poly(\log n, 1/\eps)$.
\end{thm}
\begin{proof}[Proof]
	By Newman's theorem, we can assume that Alice and Bob share randomness. Similarly, by the preceding observations we assume they know $\max_{i\in A\cup B} f(i)$ and a value $V$ satisfying $V\in [(1-\eps)f(OPT),\;f(OPT)]$.
	
	The high-level idea is for Alice and Bob to simulate together the Approximate-or-Separate \Cref{alg:Approximate-or-Separate} using as its cutting plane method (a variant of) the algorithm of \cite{bhattacharya2023chasing} used in \Cref{lem-main-fractional-fast}. 
	Alice and Bob will together compute a vector sequence $\vec{x}^0=\vec{0},\vec{x}^1,\dots,\vec{x}^R$ chasing the $R=\tilde{O}(r)$ halfspaces during the run of the polynomial-time submodular objective chasing algorithm of \Cref{lem-main-fractional-fast} when applied to two instances: the first instance has $E_0=\emptyset$, and the second instance has $E_1=A\cup B$.
	Since the optimal offline algorithm has recourse at most $r$ on this trivial sequence, by the competitive recourse of our submodular objective chasing algorithm (see \Cref{lem-main-fractional-fast}), the absolute recourse of the algorithm Alice and Bob simulate is $\tilde{O}_\eps(r).$
	At each time $t$ of the joint execution of the algorithm, Alice will know all $x^t_i$ for all $i\in A$ and similarly Bob will know $x^t_i$ for all $i\in B$, which is trivial at first, as $\vec{x}^0=\vec{0}$.

	\paragraph{Sampling sets: from low recourse to low communication.} Initially,  Alice and Bob sample (together) a set $W$ of $O(\eps^{-2}\log n)$ many i.i.d.~``witness times'' $w \sim \Uni[0,1]$. They also maintain together for each $w\in W$ a set $S^t(w)$ containing each element $i\in A\cup B$ independently with probability $1-\exp(-x^t_i)$.
	Alice and Bob both know $S^t(w)$ at time $t$ for each $w\in W$.
	In order to avoid resampling and communicating this set at each time $t$, they initially sample variables $Z_i^w\sim \Uni[0,1]$ for each $i\in A\cup B$, and include $i$ in $S^t(w)$, if and only if $1-\exp(-x^t_i) \geq Z^t_i$, which by happens with probability $1-\exp(-x^t_i)$, and so these sets are drawn as in \Cref{prop-mainsep}, as needed by the algorithm of \Cref{sec:chasing-optimally}, and we use shortly hereadter. Crucially, since the total recourse $\sum_t \|\vec{x}^t-\vec{x}^{t-1}\| = \tilde{O}(r)$, this implies that Alice and Bob can simply notify of all $i\in S^t(w)\oplus S^{t-1}(w)$. But this event happens for each $i$ at time $t$ when $1-\exp(-x^t_i)<Z^t_i \leq 1-\exp(-x^{t-1}_i)$ or $1-\exp(-x^{t-1}_i)<Z^t_i \leq 1-\exp(-x^{t}_i)$, which happens with probability  
	$$\left|1-\exp(-x^t_i)-1-\exp(-x^{t-1}_i)\right|\leq \left|x^t_{i}-x^{t-1}_i\right|,$$
	where the inequality is a simple corollary of the intermediate value theorem.
	But then, by the recourse bound and linearity of expectation, the total number of elements $(O(\log n)$ bits each) which Alice and Bob need to communicate to each other to maintain these sets $S^t(w)$ is $\tilde{O}(r)$.
	
	\paragraph{Separation: from few rounds to low communication.}
	We now show how to implement each of the $\tilde{O}(r)$ rounds of the algorithm of \Cref{lem-main-fractional-fast}. (To see this bound on the number of rounds: recall that this algorithm incurred $\Omega(\eps)$ recourse per violated constraint found, and so the bound on the number of rounds/cutting planes follows from the preceding recourse upper bound.)
	To implement this algorithm, we need to (1) find a violated constraint in the Wolsey polytope, if one exists, (2) find a partition constraint which is $(1+\eps)$ violated, if one exists (this parallels the very violated constraint in the matroid polytope paragraph in \Cref{lem-main-fractional-fast}), and (3) update $\vec{x}$. We show how to implement each of these using $\tilde{O}(r)$ bits over all rounds: For (1), since Alice and Bob both know $S^t(w)$ for all $w\in W$, they can each compute $f(i\mid S^t(w))$ for each element $i$ in their set, and communicate $\sum_{i\in A}f(i\mid S^t(w))\cdot x_i$ and $\sum_{i\in B}f(i\mid S^t(w))\cdot x_i$ to each other, allowing them to compute $f(S)+\sum_{i\in A\cup B}f(i\mid S^t(w))\cdot x_i$.
	By our observation that Alice and Bob need only communicate marginals of value at least $(\eps/n)\cdot \max_{i\in A\cup B} f(i)$, and since they know $\max_{i\in A\cup B} f(i)$ by \Cref{obs:guess-opt}, this requires only $O_\eps(\log n)$ bits per round to obtain this value up to additive error of $\eps$, and so this requires $\tilde{O}(r)$ bits to communicate overall. 
	For (2), in each part $P$ we have Alice notify Bob when $\sum_{i\in A\cap P}x_i - $ increases or decreases by at least $\eps/4$ from the last time a message concerning part $P$ was sent from Alice, in which cases Alice notifies Bob of the closest multiple of $\eps/4$ that this sum attained. Since there are $4n/\eps$ such values, this requires sending $O_\eps(\log n)$ (accounting for $O(\log r)=O(\log n)$ bits to encode the index $P$) whenever Alice notifies Bob.
	But since each such message requires $\Omega(\eps)$ recourse, Alice sends at most $\tilde{O}_\eps(r)$ such messages throughout, and similarly for Bob.
	Therefore, Alice and Bob know $\sum_{i\in P} x^t_i$ up to an additive $\eps/2$.
	Whenever their estimate exceeds the cardinality constraint for the part by $\eps$, 
	Finally, for (3), which as discussed in \Cref{lem-main-fractional-fast} reduces to a binary search, Alice and Bob simulate the binary search to compute an approximation of the update step of \cite{bhattacharya2023chasing}, which as discussed in \Cref{lem-main-fractional-fast} suffices for the algorithm's correctness.
	Each such round requires Alice and Bob to compute their contribution to the KL divergence and send this value to Bob and Alice, which can be done using $O_\eps(\log n)$ many bits while retaining the required precision.
	Over the $\tilde{O}(r)$ time steps $t$, this requires Alice and Bob to send $\tilde{O}(r)\cdot O_\eps(\log^2n) = \tilde{O}(r)$ many bits.
	
	\paragraph{Rounding:} Finally, it remains to round the solution. Alice and Bob together perform pivotal sampling in each part \cite{srinivasan2001distributions}, as follows: Alice first performs this rounding for her elements, and stops before rounding the last element in each part $P$ with $\sum_{i\in A\cap P} x_i$ is not integral: thus, for each part $P$ in the partition, some elements in $A\cap P$ have their $x^R_i$-value rounded to zero, $\lfloor \sum_{i\in A\cap P} x^R_i\rfloor$ elements have their $x$-value rounded to one (in which case she sends them to Bob), and one element retains a fractional value $\sum_{i\in A\cap P} x_i - \lfloor \sum_{i\in A\cap P} x_i\rfloor$. Alice then sends at most $1+\lfloor \sum_{i\in A\cap P} x_i\rfloor$ IDs to Bob, including the fractional value $\tilde{x}^t_i$ of the last element $i$ in the part. 
	As $\sum_i x^R_i\leq r$, at most $r+\sum_i x_i = \leq 2r$ IDs are sent, using $O(r\log n)$ bits.
	To make sending the fractional values cheap,  the fractional part is  rounded down to the closest multiple of $1-\eps$ (decreasing the multilinear value by at most $1-\eps$, by \Cref{fact:ML-scaling}) and further rounding down to zero if this new value is less than $\eps/n$ (which by linearity decreases the multilinear objectives by at most $\eps/n\sum_i x_i \cdot \max_j f(j) \leq \eps\cdot f(OPT)$). This value therefore requires $O(\log_{1/(1-\eps)}\log(n/\eps))=O_\eps(\log n)$ many bits to transmit. 
	Following this, Bob obtains a subset $S$ of $[n]$ with each element $i$ belonging to it with probability $x_i$ (up to the negligible rounding of values discussed above) and this distribution is NA (see \cite{branden2012negative,dubhashi2007positive,byrka2018proportional}). This then implies that $\E[f(S)]\geq F(\vec{x})$, by \cite{christofides2004connection,qiu2022submodular}, as desired.
\end{proof}
\section{Summary and Open Questions}\label{sec:summary}

In this work we introduce the chasing submodular objectives problem, and provide polynomial-time algorithms with optimal approximation and competitive recourse for the two most widely-studied basic constraints: cardinality constraints and partition matroids. To obtain polynomial-time algorithm (for this and other models) we introduce a new approximate-or-separate framework, allowing us to reduce constrained submodular maximization to any cutting plane method.

\paragraph{Chasing subomdular objectives.} A natural question we leave open is the existence of algorithms with optimal, or even $\Omega(1)$-approximate, submodular objective chasing subject to richer constraint families. 
We note that for the expressive family of matroid constraints, we provide an optimal fractional algorithm in \Cref{sec:chasing-optimally}, which reduces this search to that of designing approximate recourse-respecting randomized rounding algorithms.
Another natural questions is the study of submodular objective chasing with \emph{non-monotone} submodular functions, which while less common in applications, are also theoretically interesting. 
One challenge is that for such functions $f^*$ is not an extension (and also does not yield a packing-covering LP), and so other ideas are~needed.

\paragraph{Submodular Maximization.}
In this work we introduce the approximate-or-separate method. This method allows us to obtain optimally-approximate (fractional) solutions to submodular maximization subject to separable constraints, using nothing more than a separation oracle and cutting plane methods.
We illustrate this method's utility for several models of computation: static algorithms with curvature-sensitive approximation, and communication complexity protocols.
In follow-up work, \cite{naor2026dimension} use our approximate-or-separate framework to compute \emph{several sets} of high submodular value, in a multi-scenario reallocation scenario similar; they do so by extending our approach, using the ellipsoid method and writing an LP imposing multiple high submodular value  constraints (for parts of the output vector), which they then round using a novel correlated sampling for the hypersimplex.
We anticipate further applications of our method, and leave it as an open problem to find further models in which this meta-algorithm (perhaps using low-recourse cutting plane methods, as we do in \Cref{sec:communication-complexity}) can find applications to other computational models.

\appendix
\section*{APPENDIX}
\section{Competitive Recourse Lower Bounds}\label{app:lower-bounds}

In this section we (re-)state and prove the simple competitive recourse lower bounds stated in \Cref{sec:intro}.

\chasingLB*
\begin{obs}\label{lem:det-chasing-LB}
    No \emph{deterministic}  constant-approximate submodular objective chasing algorithm has $o(\max_t |E_t|)$-competitive recourse, even for a fixed coverage function $f$ and a cardinality constraint of $k=1$.
\end{obs}

\begin{proof}[Proof (of both observations)]
    Consider an input with $n$ different elements arriving, all then departing in a random order. Consider the submodular function $f(S)\triangleq \min\{1,|S|\}$. (This is a degenerate coverage function.)
   The offline optimal algorithm knows the last element to be deleted, and thus pays a recourse cost of $2$.
    In contrast, when $t$ elements remain, for any constant $c$, a $c$-approximate (randomized) online algorithm must have a solution of expected size at least $c$.
    So, the probability that one of the items stored by the algorithm is deleted in time $t$ is $c/t$.
    The total expected numbers of elements ever stored by the algorithm and deleted is therefore at least $\sum_{t=1}c/t = \Omega(\ln n)$, where indeed $n=\max_t |E^t|$.

     The same construction, with the next element deleted belonging to the current solution of the deterministic algorithm, forces the number of elements deleted (and hence the recourse) to be~$n$. That is, it forces exponentially higher competitive recourse of $n/2=\Omega(\max_t |E_t|)$.
\end{proof}

\section{Absolute Recourse in Sequences of Interest}\label{sec:special-cases}

In this section we focus on the case of fixed objective function $f$ and cardinality constraints for sequences of interest: incremental, decremental, and sliding windows. We show that for these, logarithmic amortized absolute recourse suffices even to chase a $(1-\eps)$ approximation of the optimal value $OPT_t$ at each time $t$. (Consequently, our submodular value chasing algorithms of \Cref{sec:chasing-optimally} achieve polylogarithmic absolute recourse with $(1-1/e-O(\eps))$ approximation of any $V_t\leq OPT_t$.)

\paragraph{Partially-dynamic sequences.} We start with incremental or decremental settings over $T$ steps, where $|E_t\oplus E_{t-1}|=1$ for all $t\in [T]$ and $E_t$ either only grows over time (incremental, insertion-only) or shrinks over time (decremental, deletion only), where in the latter case the input starts with $T$ elements, all of which are eventually deleted.

It is known that in the worst case, in the most general fully-dynamic setting where at each time $t$ we have $E_t=E_{t-1}\cup \{e_t\}$ or for each time $t$ we have $E_t=E_{t-1}\setminus \{e_t\}$, sufficiently high approximation requires absolute worst-case recourse to essentially be as high as recomputing a solution at each step, even for fixed $f$ and cardinality constraints \cite{dutting2024consistent,dutting2025cost}.
In this section we show that for $(1-\eps)$-approximation of $OPT_t$, \emph{amortized} absolute recourse is much smaller for insertion-only or deletion-only streams, where $E_t$ only grows (by one) at each time step, or shrinks (by one) at each time step. We start with a more general statement with a dependence on the aspect ratio, $\Delta\triangleq \max_S f(S)/\min_{S: f(S)\neq 0} f(S)$.

\begin{obs}\label{obs:ins/del-amortized-recourse}
    Let $f:2^E\to \mathbb{R}_+$ be a  function with non-zero values in the range $[1,\Delta]$. Then, there exists an algorithm for chasing $f$ with $V_t=(1-\eps)\cdot OPT_t$ for all $t$ (subject to any constraint) under insertion-only or deletion-only streams with amortized absolute recourse $O(\eps^{-1}\log \Delta)$.
\end{obs}
\begin{proof}
    The algorithm is as follows: Whenever $\max_{S\subseteq A} f(S)$, the optimal value restricted to the currently \emph{available} items $A$, changes by a factor of $1\pm \epsilon$, we change the solution completely, using at most $T$ changes to the solution. Since this maximum changes by a $1+\epsilon$ factor only $\log_{1+\epsilon}(\Delta) = O(\eps^{-1}\log \Delta)$ many times, by the range of values and monotonicity of the objective in partially-dynamic settings, the total recourse is $O(T \eps^{-1}\log \Delta)$, as claimed.
\end{proof}

For cardinality constraint $k$, the above dependence on $\Delta$ is suboptimal, and one can replace it by (effectively) only a dependence on $k$, provided $f$ is subdmodular, or even sub-additive, i.e., if $f(S)\leq \sum_{i\in S} f(i)$ for all $S\subseteq E$.
\begin{lem}\label{lem:info-theoretic}
     Let $f:2^E\to \mathbb{R}_+$ be a non-negative sub-additive function. Then,  there exists an algorithm for
     chasing $f$ with $V_t=(1-\eps)\cdot OPT_t$ subject to cardinality constraint $k$ in insertion-only streams using amortized absolute recourse $O(\eps^{-1}\log (k\eps^{-1}))$.
\end{lem}
\begin{proof}
   Denote by $A$ the set of \emph{available} items, and by $L(A):=\{e\in A \mid f(e)\geq \frac{\eps}{k}\cdot \max_{e\in A}f(e)\}$ the \emph{large} items in $A$. The algorithm works as follows: whenever $\max\{f(S) \mid S\subseteq A, |S|\leq k\}$ grows by a $1/(1-\epsilon)$ factor, we change the solution to be some $S\in \arg\max\{f(S) \mid S\subseteq L(A), |S|\leq k\}$.
   We call each such change to the algorithm's solution the beginning of a \emph{phase}.

    To analyze this algorithm, 
    we rely on the following corollary of sub-additivity:
    \begin{align}\label{eqn:large-only}
    \max_{\substack{S\subseteq L(A) \\ |S|\leq k}} f(S) \geq (1-\eps)\cdot \max_{\substack{S\subseteq A \\ |S|\leq k}} f(S).
    \end{align}
    To see this, we fix a set $S\in \arg\max\{f(S)\mid S\subseteq A, 
    |S|\leq k\}$ and repeatedly invoke sub-additivity:
    \begin{align*}
        f(S) & \leq f(S\cap L(A)) + f(S\setminus L(A)) \\ 
        & \leq f(S\cap L(A)) + \sum_{e\in S\setminus L(A)} f(e) \\
        & \leq f(S\cap L(A)) + k\cdot \frac{\eps}{k} \max_{e\in A}f(e) \\
        & \leq f(S\cap L(A)) + \eps\cdot  \max_{S\subseteq  A}f(S).
    \end{align*}

    Denoting by $A$ and $A'$ the set of active elements at the beginning and middle point of a phase:
    \begin{align*}
        \max_{\substack{S\subseteq L(A) \\ |S|\leq k}} f(S) \overset{\eqref{eqn:large-only}}{\geq} (1-\epsilon) \max_{\substack{S\subseteq A \\ |S|\leq k}} f(S) \geq (1-\epsilon)^2 \max_{\substack{S\subseteq A' \\ |S|\leq k}} f(S),   
    \end{align*}
    That is, the algorithm is $(1-\epsilon)^2\geq (1-2\epsilon)$-approximate at all times. 
    
    It remains to analyze the algorithm's absolute recourse.
    For this, we partition phases into \emph{eras} of $\log_{1/(1-\epsilon)}(k^2\eps^{-1}) + 1= O(\eps^{-1}\log (k\eps^{-1}))$ many consecutive phases. Letting $A$ and $A'$ be the available items at the beginning of an era and the next era, we have that
    \begin{align*}
        k\cdot \max_{e\in A'} f(e) \geq \max_{\substack{S \subseteq A' \\ |S|\leq k}} f(S) > \frac{k^2}{\eps}\cdot  \max_{\substack{S \subseteq A \\ |S|\leq k}} f(S) \geq \frac{k^2}{\epsilon} \cdot \max_{e\in A} f(e).
    \end{align*}
    Re-arranging, we have $\frac{\eps}{k}\cdot \max_{e\in A'}f(e) > \max_{e\in A} f(e)$. Therefore, none of the large items of the $i^{th}$ era belong to the $(i-2)^{th}$ era. 
    Consequently, in era $i$ the algorithm only keeps solutions containing items inserted in eras $i$ and $i-1$ that were not small during their insertion.
    Letting $T_i$ be the number of items inserted in era $i$ that were large upon insertion, we have that each solution used by the algorithm in era $i$ contains at most $T_i+T_{i-1}$ many items. (We define $T_0=0$.) Consequently, since each era consists of $O(\eps^{-1}\log(k\eps^{-1})$ many phases and solutions maintained by the algorithm, the latter's total recourse is, as claimed 
    \begin{align*}
        O(\eps^{-1} \log(k\eps^{-1}))\cdot \sum_{i} (T_i+T_{i-1})  & \leq O(\eps^{-1} \log(k\eps^{-1}))\cdot 2\sum_i T_i \leq O(T\cdot \eps^{-1} \log(k\eps^{-1}).\qedhere 
    \end{align*}
\end{proof}

\begin{rem}
    By symmetry, recomputing when $OPT$ grows smaller by a $1-\eps$ factor chases $f$ with value $(1-\eps)\cdot OPT_t$ and recourse for decremental settings ending with an empty input.
\end{rem}


The preceding argument can be extended to obtain a $(1-1/e-\eps)$-approximate (polynomial-time) submodular objective chasing algorithm.

\begin{lem}
     For any monotone non-negative submodular function $f:2^E\to \mathbb{R}_+$, there exists an algorithm for chasing $f$ with $V_t=(1-1/e-\eps)\cdot OPT_t$ subject to any matroid constraint in insertion-only (and deletion-only) streams using amortized absolute recourse $O(\eps^{-1}\log (k\eps^{-1}))$ \emph{and polynomial time per update}.
\end{lem}
\begin{proof}[Proof (sketch)]
    The proof is essentially the same as that of \Cref{lem:info-theoretic}, but has us recompute when the output of some poly-time $(1-1/e-\eps)$-approximate algorithm for monotone submodular maximization subject to a matroid constraint (e.g., \cite{calinescu2011maximizing}) has its output's value increase (for insertion-only streams) or decrease (for deletion-only streams) by a $1+ \eps$ factor.
\end{proof}

\paragraph{Sliding Windows.}
We now show that low absolute recourse is attainable for a well-studied setting with both insertions and deletions, namely the \emph{sliding window} model.
Here, elements are revealed over time, and we must maintain a solution over the last $L$ elements.  Alternatively, elements are added and removed in FIFO order, and when the queue has length $L$, and after it becomes full at each time step one element is added to the back of the queue and another leaves from the head of the queue.


Here we adapt the smooth/exponential histogram technique from the (sliding-window) streaming literature \cite{datar2002maintaining,braverman2007smooth}, usually used to obtain low-space algorithms, to prove the following low-recourse result.

\begin{lem}\label{lem:info-theoretic-sliding-window}
     Let $f:2^E\to \mathbb{R}_+$ be a non-negative submodular function. Then,  there exists a $(1-\epsilon)$-approximate algorithm for maximizing $f$ subject to cardinality constraint $k$ under sliding window streams of length $L$  using amortized absolute recourse $O\left(\eps^{-1}\log (k\eps^{-1})\right)$.
\end{lem}
\begin{proof}
    We maintain a list of ``anchors'', $s_1<s_2<\dots<s_m$, where $s_1\geq t-L+1$ and $s_m\leq t$, where $g_t(s_i) := f(OPT[s_i,t])$ satisfies
    $$g_t(s_i) \geq (1-\eps) \cdot g_t(s_{i-1}).$$
    Moreover, whenever three consecutive anchors $s_{i-1},s_i,s_{i+1}$ satisfy
    $$g_t(s_{i+1})\geq (1-\eps)\cdot g_t(s_{i-1}),$$
    we delete anchor $s_i$ (and relabel the other anchors).
    We only create an anchor if its sub-window contains an element of singleton value at least $\frac{\eps}{k}\cdot \max \{f(e_i) \mid i\in [t-L+1,t]\}$, noting that dropping these elements only incurs an additive $\eps\cdot \max\{f(e_i)\mid i\in [t-L+1,t]\}\leq \eps\cdot f(OPT[t-L+1,t])$ loss.
    The multiplicative gaps between anchors' window values guarantees a $1-\eps$ multiplicative gap between nonconsecutive anchors' $g_t(\cdot)$ value, and so the number of anchors $m$ is bounded at all times $t$: 
    $$m:=O\left(\log_{1+\Theta(\eps)}\left(\frac{\max_i g_t(s_i)}{\min_i g_t(s_i)}\right)\right)=O\left(\eps^{-1} \log \left(\frac{\max_i g_t(s_i)}{\min_i g_t(s_i)}\right)\right)=O(\eps^{-1}\log(k/\eps)).$$
    Above, the last step uses that by the preceding discussion we can assume that each element has singleton value at least $\frac{\eps}{k}\cdot \max \{f(e_i) \mid i\in [t-L+1,t]\}$, where we use the set of available items (as in the proof of \Cref{lem:info-theoretic}).
    Note that following the addition of the next element, we have (by submodularity), that if $g_t(s_{i+1})\geq (1-\eps)\cdot g_t(s_i)$ and $s_i\geq t-L+2$, then 
    $g_{t+1}(s_{i+1})\geq (1-\eps)\cdot g_{t+1}(s_i)$, since the addition of element $t+1$ to the end of the windows with anchors $s_{i+1}$ and $s_i$ increases the maximum value of their windows more for the smaller window (submodularity $OPT(A)$) 

    At any point in time, our output solution is the solution of value $g_t(s_1)$, for $s_1$ the earliest anchor in $[t-L+1,t]$.
    Since by the preceding argument this solution has value at least $(1-\eps)\cdot g_t(s_0) \geq f(OPT[t-L+1,t]$, for $s_0$ the last anchor to leave the window, this is a $(1-\eps)$-approximation.
    Note that from the moment $s_1$ becomes the earliest anchor, this is effectively an incremental algorithm, and so the recourse incurred is $O(k)$ (for the first solution for this incremental algorithm), plus $O(\eps^{-1}\log(k/\eps))$ amortized per update, due to \Cref{lem:info-theoretic}.
    Since for each window of length $L$ at most $m$ values can ever be the first (non-deleted) anchor,
    the above implies an amortized recourse of 
    \begin{align*}
        O\left(\frac{mk}{L}+\eps^{-1}\log(k/\eps)\right) = O(\eps^{-1}\log(k/\eps)),
    \end{align*}
    using that $L\geq k$ (else we can take all elements in the window, using  $O(1)$ recourse per step).
\end{proof}


\section{Alternative Proof of \texorpdfstring{\Cref{prop-mainsep}}{}: Approximate-or-Separate}\label{sec:deferred-fractional}

In this section we prove a version of \cref{prop-mainsep} with slightly poorer constants (although asymptotically equivalent to the above).

\begin{restatable}{prop}{separateorapproximate}
    Let $f\colon 2^{E} \rightarrow \R_+$ be a non-negative monotone submodular function, $\vec{x}\in [0,1]^n$, let $\eps\in (0,1)$ and $V> 0$. If $F(\vec{1}-\exp(-\vec{x}))<(1-1/e-\eps) \cdot V$, then  
    a set $S \triangleq \{i\in E \mid  Y_i \leq t\}$, for independent $t\sim \Uni[0,1]$ and $Y_i\sim \Exp(x_i)$ for all $i\in E$, satisfies  with probability at least~$\Omega(\eps^2)$:   \begin{equation*}
        f(S)+ \sum_{i\in E}f(i\mid S) \cdot x_i \leq \left(1-\Omega(\eps) \right) \cdot V.
    \end{equation*}
\end{restatable}
\begin{proof}
For a fixed $t\in[0,1]$,
define the random set $S(t)\triangleq\{i \mid Y_i \leq  t\}$. 
We first observe that up to $O(dt^2)$ terms for any $t\in[0,1]$ and set $S\subseteq E$:
\begin{align*}
    \E[f(S(t+dt))-f(S(t)) \mid S(t)=S] & = \bigg[\sum_{i\in E}f(i \mid S) \cdot x_i \bigg]dt.
\end{align*}
Let $\phi(t)=\E[f(S(t))]$, where expectation is over the exponential variables $Y_1,\dots, Y_n$.
By taking total expectation over the set $S(t)$, we get that in the limit as $dt$ tends to zero:
\begin{align*}
    \phi(t) + \frac{d \phi(t)}{dt} & = \E[f(S(t)] + \frac{1}{dt}\E[f(S(t+dt)-f(S(t)] \\
& = \sum_{S\subseteq E}\bigg[f(S)+\sum_{i\in E}f(i \mid S) \cdot x_i \bigg]\cdot Pr[S(t)=S)].
\end{align*}
Therefore, by Markov's inequality, for any fixed $t\in[0,1]$ that satisfies $\phi(t) + \frac{d \phi(t)}{dt} \leq V(1-\frac{\eps}{2e})$,  
\begin{align}\label{eqn:wolsey-unlikely-high}
\Pr_{S\sim S(t)}\bigg[f(S)+\sum_{i\in E}f(i \mid S) \cdot x_i & \geq V\left(1-\frac{\eps}{4e}\right)\bigg] \leq 
\frac{1-\frac{\eps}{2e}}{1-\frac{\eps}{4e}} \leq 1-\frac{\eps}{4e}.\end{align}

We now prove that the precondition of \Cref{eqn:wolsey-unlikely-high} holds with non-trivial probability:
\begin{equation}\label{eqn:expected-wolsey-low}
   \Pr_{t \sim \Uni[0,1]}\left[\phi(t) + \frac{d \phi(t)}{dt} \leq V\left(1-\frac{\eps}{2e}\right)\right] \geq \frac{\eps}{2}. 
\end{equation}

Assume to the contrary that this inequality does not hold. Then, 
\begin{align*}
    \Pr_{t \sim \Uni[0,1]}\bigg[e^{t} V - \frac{d (e^{t}\phi(t))}{dt} \geq  \frac{\eps V}{2}\bigg]  & \leq \Pr_{t \sim \Uni[0,1]}\bigg[e^{t} V - \frac{d (e^{t}\phi(t))}{dt} \geq e^{t} V \frac{\eps}{2e}\bigg] \\
    & =\Pr_{t \sim \Uni[0,1]}\bigg[e^{t}\left(\phi(t) + \frac{d \phi(t)}{dt}\right) \leq e^{t} V\left(1-\frac{\eps}{2e}\right)\bigg]\\ & = \Pr_{t \sim \Uni[0,1]}\bigg[\phi(t) + \frac{d \phi(t)}{dt} \leq V\left(1-\frac{\eps}{2e}\right)\bigg]  \\
    & \leq \frac{\eps}{2}.
\end{align*}

As $e^{t} V - \frac{d (e^{t}\phi(t))}{dt}\leq eV$ for $t\in [0,1]$, we get
$\E_{t \sim \Uni[0,1]}\left[e^{t} V - \frac{d (e^{t}\phi(t))}{dt}\right] \leq (1-\frac{\eps}{2}) \cdot \frac{\eps V}{2} +  \frac{\eps}{2} \cdot eV \leq \eps e V$.
Hence, 
\begin{align*}
 \eps e V & \geq \E_{t \sim \Uni[0,1]}\bigg[e^{t} V - \frac{d (e^{t}\phi(t))}{dt}\bigg] \\
& = \int_{t=0}^{1}\bigg[e^{t} V - \frac{d (e^{t}\phi(t))}{dt}\bigg]dt \\
& = (e-1) V - e \cdot \phi(1) -\phi(0)\\
 & = (e-1) V - e \cdot \E_{S\sim S(1)}[f(S(1))],
\end{align*}
where the last step relied on $f$ being normalized, and hence $\phi(0) = f(S(0))=f(\emptyset)=0.$

Rearranging, we get $\E[f(S(1))] \geq (1-1/e-\eps)V$. However, letting 
$\vec{z} \triangleq \vec{1}-e^{-\vec{x}}$ be a vector with $z_i= 1-e^{-x_i} = \Pr[Y_i \leq 1]$, this contradicts the proposition's hypothesis, $F(\vec{z}) < (1-1/e-\eps)V$, since $\E[f(S(1))] = F(\vec{z})$.

Finally, since \Cref{eqn:expected-wolsey-low} implies the precondition of \Cref{eqn:wolsey-unlikely-high} holds with probability at least $\frac{\eps}{2}$, we find that for $t\sim \Uni[0,1]$, with probability at least $\frac{\eps}{4e} \cdot \frac{\eps}{2}=\Omega(\eps^2)$, the set $S=S(t)$ satisfies $f(S)+ \sum_{e}f(i \mid S) \cdot x_i \leq V(1-\frac{\eps}{4e})$, as desired.
\end{proof}

Finally, we now substantiate  the claim that a set of $O(\eps^{-1}\log n)$ many guesses suffice to guess $f(OPT)$ with a $(1-\eps)$ factor, yielding the following, together with our approximate-or-separate \Cref{alg:Approximate-or-Separate}: 
\guessing*
\begin{proof}
Let $OPT\triangleq \max_{\vec{x}\in \calP}F(\vec{x})$. We assume that $\vec{e}_i \in\calP$ for all $i\in E$. Then, by submodularity, we have that $OPT\in [\max_{i\in E} f(i), n\cdot \max_{i\in E} f(i)]$. We look at all values $V_k=(1-\eps)^k$ in the range $[\max_{i\in E} f(i), n\cdot \max_{i\in E} f(i)]$ (there are $N\triangleq O(\eps^{-1}\cdot {\log n})$ such values). Suppose for the moment that whenever \Cref{alg:Approximate-or-Separate} is executed with an input value value $V_k$ and outputs a vector $\vec{x}$, the vector is (deterministically) always in $\calP$ and also always satisfies $F(\vec{x})\geq F(\vec{1}-\exp(-\vec{x})) \geq  (1-1/e-\eps) \cdot V_k$. In this case, we are guaranteed that whenever $V_k\leq OPT\leq \max_{x\in P}f^*(x)$, the algorithm always outputs such a vector, and for each $V_k>OPT$, for which $\calK(V_k)$ may be empty, the algorithm may or may not output such a vector.

We perform a binary search on the values $V_k$ and output a vector with a maximal value for which the algorithm returns such a vector. We are guaranteed that this value is at least $\max_{k}\{V_k\leq OPT\}\geq OPT(1-\eps)$. This requires $O(\log N) = O(\log (\eps^{-1}\cdot \log n))$ executions of the Approximate-or-Separate algorithm. Finally, as the algorithm may fail with probability $\delta$ on each of these executions, we set the probability of failure $\delta$ to be $O(\frac{\eps}{\log N})$. This guarantees by a union bound that the algorithm succeeds with probability $1-O(\eps)$ on all  values in the binary search.
If the algorithm fails on any of the values (which happens with probability $O(\eps)$), we may return $\vec{x}$ with no guarantee on the value $F(\vec{x})$. On the other hand, by \Cref{thm-fractional1}, whenever the algorithm outputs a vector, then $\vec{x} \in \calP$ satisfies $F(\vec{x})\geq F(\vec{1}-\exp(-\vec{x})) \geq (1-1/e-\eps) \cdot V$ with probability $1-\delta$, and hence $\E[F(\vec{x})]\geq (1-1/e-O(\eps)) \cdot V$.  
Overall, ignoring logarithmic terms, the number of oracle queries to $\calP$ is $O(T\cdot \log N) = \tilde{O}(T)$, and the number of oracle queries to $f$ is $\tilde{O}({n \cdot T}\cdot {\eps^{-2}})$. 
\end{proof}
\section{Recourse-Respecting Randomized Rounding for \texorpdfstring{$k=1$}{}}\label{app:unit-capacity}
    
    In this section we present some recourse-respecting randomized rounding, starting with the first one, due to 
    \cite{keyfitz1951sampling}, 
    and then another algorithm requiring only sampling a single $\Uni[0,1]$ variable.
    We claim no novelty for these algorithms, and present them mostly for completeness.
    
    We note that for both algorithms, negative association follows directly from Property \ref{prop:k-uniform} and \Cref{lem:na}, \ref{lem:0-1-lemma}. Similarly, the submodular value follows then from Properties \ref{prop:marginals} and \ref{prop:na} together with \Cref{lem:NA->submod} and monotonicity of $f$. We therefore only prove Properties \ref{prop:k-uniform}-\ref{prop:marginals} for these algorithms.

    \paragraph{Keyfitz's Algorithm.} 
    The algorithm is basically rejection sampling with corrections. 
    By adding a dummy coordinate, we assume that all points are in the probability simplex, namely $\norm{\vec{x}^{t}}_1=1$ for all $t$.
    Sampling from $\vec{x}^1$ is trivial. We describe the step of this algorithm for consecutive vectors $\vec{x}=\vec{x}^{t}$ and $\vec{y}=\vec{x}^{t+1}$, where we assume by induction on $t\geq 1$ that a set $X$ drawn according to $\vec{x}$ is already given to us.

    \begin{algorithm}[H]
    	\caption{Keyfitz's Online Rejection Sampling}
    	\label{alg:keyfitz-rrrr}
     \begin{algorithmic}[1]
        \Ensure Set $X\subseteq [n]$ with $|X|= 1$ drawn according to $\vec{x}$
        \If{$X=\{i\}$ \textbf{ and } $x_i\leq y_i$} \State Set $Y\gets X$.
        \Else 
        \State Set $Y\gets \{j\}$ for some element $j$ with $y_j>x_j$ chosen with probability $\frac{y_j-x_j}{\sum_{j: y_j>x_j} (y_j-x_j)}$. 
        \EndIf
        \State \textbf{Return} $Y$.
    \end{algorithmic}	
    \end{algorithm}	
    
    \begin{lem}
        \Cref{alg:keyfitz-rrrr} applied to a vector sequence $\vec{x}^{1},\vec{x}^2,\dots$ with $\norm{\vec{x}^t}\leq 1$ for all $t$ outputs a sequence of sets $X_1,X_2,\dots$ with $|X_1|\leq 1$, each random set $X^t$ having an NA  characteristic vector with marginals $\vec{x}^t$, using recourse $\E[\; | X_t \oplus X_{t-1}|\; ]\leq \norm{\vec{x}^t-\vec{x}^{t-1}}_1$.
    \end{lem}
    \begin{proof}
        Property \ref{prop:k-uniform}, i.e., that \Cref{alg:keyfitz-rrrr} outputs sets of cardinality $k=1$ is trivial. Marginals (Property \ref{prop:marginals}) are also easy to check by induction on $t$ (the base case being trivial): the input $X$ satisfies them by induction; for $Y$: if $x_i > y_i$, then $\Pr[i\in Y]=\frac{y_i}{x_i}\cdot x_i = y_i$. Else, $x_i\leq y_i$, and 
        $$\Pr[i\in Y] = x_i + \left(\sum_{j: x_j>y_j} \frac{x_j-y_j}{x_j}\cdot x_j\right)\cdot \frac{y_i-x_i}{\sum_{j: y_j>x_j}(y_j-x_j)} = x_i + (y_i-x_i) = y_i,$$
        where we used that $\sum_{j: x_j>y_j} (x_j-y_j) = \sum_{j: x_j<y_j} (y_j-x_j)$.
        Finally, for the recourse (Property \ref{prop:distance}), we have (noting that the dummy values do not increase the $\ell_1$ distance between $\vec{x}$ and $\vec{y}$):
        \begin{align*}
        \E[|X\oplus Y|] = 2\cdot \Pr[X\neq Y] & = 2\cdot \sum_{i: x_i>y_i} \frac{x_i-y_i}{x_i}\cdot x_i =  2\cdot \sum_{i: x_i>y_i} (x_i-y_i) = \norm{\vec{x}-\vec{y}}_1. \qedhere 
        \end{align*}
    \end{proof}

\paragraph{An alternative algorithm for $k=1$.}

This algorithm is simple to state intuitively: we throw a dart at the unit interval~$[0,1]$. At each time $t$, we maintain families of disjoint intervals corresponding to the elements of $E$, where ${\cal F}_i$ denotes the intervals of element $i\in E$, which have total measure $x^t_i$. Whenever $x_i$ decreases ($x^t_i\leq x^{t-1}_i$), we remove (parts of) intervals from $\calF_i$ of total measure $(x^{t-1}_i-x^t_i)$, and then symmetrically add some if $x_i$ increases.
The sampled element (if any) is precisely that $i$ for which the dart hit an interval in $\calF_i$. In our pseudocode (\Cref{alg:rrrr}) and subsequent discussion, we use the shorthand $\bigcup \calF \triangleq \bigcup_{I\in \calF} I$ for $\calF$ a family of intervals.

\begin{algorithm}[H]
	\caption{Recourse-Respecting Randomized Rounding}\label{alg:rrrr}
 \begin{algorithmic}[1]
    \Statex \textbf{Init:} Set $S\gets \emptyset$. Sample a $U\sim \Uni[0,1]$. Set $\calF_i \gets \emptyset$ \textbf{for all} $i\in E$.
    \For{time $t$} 
    \For{\textbf{each} $i\in E$}
    \If{$x^{t}_i\leq x^{t-1}_i$} \Comment{$\vec{x}^0=\vec{0}$ by convention}
    \State Remove (parts of) intervals of measure $(x^{t-1}_i - x^t_i)$ from $\calF_i$.
    \EndIf
    \EndFor 
    \For{\textbf{each} $i\in E$}
    \If{$x^{t}_i\geq x^{t-1}_i$} \Comment{$\vec{x}^0=\vec{0}$ by convention}
    \State Add intervals in $[0,1]\setminus \bigcup_i \calF_i$ of measure $(x^{t}_i - x^{t-1}_i)$ to $\calF_i$.
    \EndIf
    \EndFor 
    \State $S_t \gets \{i \mid \bigcup \calF_i \cap U \neq \emptyset\}.$
    \EndFor
\end{algorithmic}	
\end{algorithm}	

Let $\calF_i^t$ be $\calF$ at time $t$ (i.e., just before vector $\vec{x}^{t+1}$ is revealed). 
The following facts are easily proven by induction on $t$.
\begin{fact}\label{fact:disjoint}
    For all times $t$, the families $\calF_i^t$ are disjoint.
\end{fact}
\begin{fact}\label{fact:measure}
    For all $i\in E$ and times $t$ we have $|\calF^t_i \oplus \calF^{t-1}_i| = |x^t_i - x^{t-1}_i|$ and $|\bigcup  \calF_i|=x_i^t$. 
\end{fact}

\begin{lem}\label{thm:rrrr}
    \Cref{alg:rrrr} is a cardinality-constrained recourse-respecting randomized rounding algorithm for $k=1$.
\end{lem}
\begin{proof}
    Fix $t$. By \Cref{fact:disjoint},  $U$ belongs to at most one $\calF^t_i$, and so Property \ref{prop:k-uniform} holds ($|S_t|\leq 1$).
    Next, Property \ref{prop:distance} follows by \Cref{fact:measure}: 
    $$\Pr[i\in S_t\oplus S_{t-1}] = \Pr\left[U \in \bigcup \left(\calF^t_i \oplus \calF^{t-1}_i\right)\right] = |x^t_i-x^{t-1}_i|.$$
    Moreover, $S_t\subseteq E_t$, since by \Cref{fact:measure}, for all $i\not\in E_t$ we have $\Pr[i\in S_t] = \Pr[U \in \bigcup \calF_i] = 0$. Next, Property \ref{prop:marginals} follows by \Cref{fact:measure} and direct computation: 
    \begin{align*}\Pr\left[i\in S_t] = \Pr[U \in \bigcup \calF^t_i\right] & = x^t_i \geq 1-\exp(-x^t_i). \qedhere
    \end{align*}
\end{proof}

\section{Deferred Proofs of \texorpdfstring{\Cref{sec:curvature}}{}: Curvature Decomposition}\label{appendix:curvature}

In this section we provide a self-contained of the submodular decomposition of \Cref{lem:curvature-decomposition}, for the sake of completeness. We emphasize that this lemma is not new \cite{iyer2013curvature,iyer2013fast}, and we claim no novelty for its proof.

\curvaturedecomposition*

\begin{proof}
    Fix an ordering of set $S=\{j_1,\dots,j_m\}$ and let $S_t=\{j_1,\dots,j_t\}$ be its first $t$ elements according to this ordering.
    By submodularity, and the definition of curvature:
    \begin{align*}
        f(i\mid S_{i-1}) & \geq f(i\mid E\setminus\{i\}) \geq (1-c)\cdot f(\{i\}).
    \end{align*}
    By telescoping, this implies that $g_f=\frac{1}{c}\left(f-(1-c)\cdot \ell_f\right)$ is non-negative, since
    \begin{align*}
    f(S) = \sum_{i\in S} f(i\mid S_{i-1}) \geq \sum_{i\in S} (1-c)\cdot f(\{i\}) = (1-c)\cdot \ell(S).
    \end{align*}
    This similarly implies that 
    $g_f$ is monotone:
    \begin{align*}g_f(i\mid S) = \frac{1}{c}\left(f(i\mid S) - (1-c)\cdot \ell_f(\{i\})\right) \geq 0.
    \end{align*}
    To prove that $g_f$ is also submodular for submodular $f$, we note that for all $i$ and $S\subseteq T\subseteq E\setminus\{i\}$:
    \begin{align*}
        g_f(i\mid S) - g_f(i\mid T) & = \frac{f(i\mid S) - f(i\mid T)}{c}  \geq 0.
    \end{align*}
    Finally, to prove that $f(S)\geq g_f(S)=\frac{1}{c}\left(f(s)-(1-c)\cdot \ell_f(S)\right)$, we use the standard consequence of submodularity, $f(S)\leq \ell_f(S)$, implying
    \begin{align*}
        f(S) - g_f(S) & = \frac{1}{c}\left((1-c)\cdot(\ell_f(S) - f(S)\right) \geq 0. \qedhere
    \end{align*}
\end{proof}

\bibliographystyle{alpha}
\bibliography{abb,bib,ultimate}

\end{document}